\documentclass[prerint,11pt]{article}

\usepackage{amssymb,amsfonts,amsmath,amsthm,amscd,dsfont,mathrsfs}
\usepackage{eqsection,graphicx,float,psfrag,epsfig,color}
\usepackage{hyperref}

\newtheorem{claim}{Claim}
\newtheorem{lemma}{Lemma}

\newtheorem{theorem}{Theorem}

\newtheorem{corollary}[claim]{Corollary}

\newtheorem{note}{Note}

\definecolor{Red}{rgb}{1,0,0}
\definecolor{Blue}{rgb}{0,0,1}
\definecolor{Olive}{rgb}{0.41,0.55,0.13}
\definecolor{Green}{rgb}{0,1,0}
\definecolor{MGreen}{rgb}{0,0.8,0}
\definecolor{DGreen}{rgb}{0,0.55,0}
\definecolor{Yellow}{rgb}{1,1,0}
\definecolor{Cyan}{rgb}{0,1,1}
\definecolor{Magenta}{rgb}{1,0,1}
\definecolor{Orange}{rgb}{1,.5,0}
\definecolor{Violet}{rgb}{.5,0,.5}
\definecolor{Purple}{rgb}{.75,0,.25}
\definecolor{Brown}{rgb}{.75,.5,.25}
\definecolor{Grey}{rgb}{.5,.5,.5}
\definecolor{Pink}{rgb}{1,0,1}
\definecolor{DBrown}{rgb}{.5,.34,.16}

\definecolor{Black}{rgb}{0,0,0}

\newcommand{\no}{\nonumber}
\newcommand{\mc}{\mathcal}
\newcommand{\mb}{\mathbb}

\newcommand{\mrm}{\mathrm}
\newcommand{\f}{\frac}

\newcommand{\bs}{\backslash}
\newcommand{\pl}{\parallel}
\newcommand{\deq}{\stackrel{\text{\rm d}}{=}}
\newcommand{\diag}{\textrm{diag}}
\newcommand{\cov}{\textrm{\rm Cov}}
\newcommand{\sigal}{\mathfrak}
\newcommand{\pty}{\mc}

\newcommand{\order}{\vec{o}}
\newcommand{\indicator}{\mb{I}}
\newcommand{\almostsurely}{{\rm a.s.}}
\newcommand{\asequal}{\stackrel{\almostsurely}{=}}

\newcommand{\empr}{\hat{p}}
\newcommand{\tpsi}{\tilde{\psi}}
\newcommand{\bound}{B}
\newcommand{\pW}{p_{W}}
\newcommand{\psL}{\mathrm{PL}}
\newcommand{\tQ}{\tilde{Q}}
\newcommand{\tq}{\tilde{q}}
\newcommand{\tM}{\tilde{M}}
\newcommand{\tm}{\tilde{m}}
\newcommand{\identity}{{\sf I}}
\newcommand{\zeros}{{\sf 0}}

\def\reals{{\mathbb R}}
\def\naturals{{\mathbb N}}
\def\<{\langle}
\def\>{\rangle}
\def\E{{\mathbb E}}
\def\cC{{\cal C}}
\def\cF{{\cal F}}
\def\th{{\tilde{h}}}
\def\cP{{\cal P}}
\def\tA{\tilde{A}}
\def\normal{{\sf N}}
\def\trace{{\rm Tr}}
\def\tG{\tilde{G}}
\def\ve{\varepsilon}
\def\prob{{\mathbb P}}
\def\htau{\widehat{\tau}}
\def\dx{{\delta x}}
\def\dz{{\delta z}}
\def\va{\vec{\alpha}}
\def\vb{\vec{\beta}}
\def\lbq{\rho}
\def\lbm{\varsigma}

\def\lip{\ell}
\def\var{{\rm Var}}
\def\tV{{\tilde{V}}}
\def\slln{\varrho}

\begin{document}

\title{The dynamics of message passing on dense graphs,\\
with applications to compressed sensing}

\author{Mohsen Bayati\thanks{Department of Electrical Engineering,
Stanford University} \;\;\; and\;\;\;
Andrea Montanari${}^{*,}$\thanks{Department of Statistics,
Stanford University}}

\date{}

\maketitle


\begin{abstract}
`Approximate message passing'  algorithms have proved to be
effective in reconstructing sparse signals
from a small number of incoherent linear measurements.
Extensive numerical experiments further showed that their
dynamics is accurately tracked by a simple one-dimensional
iteration termed \emph{state evolution}.
In this paper we provide rigorous foundation to
state evolution. We prove that indeed it holds asymptotically
in the large system limit for sensing matrices with independent and
identically distributed gaussian entries.

While our focus is on message passing algorithms for compressed
sensing, the analysis extends beyond this setting, to a general
class of algorithms on dense graphs. In this context, state evolution
plays the role that density evolution has for sparse graphs.

The proof technique is fundamentally different from the
standard approach to density evolution, in that it copes
with a large number of short cycles in the underlying factor graph. It relies
instead  on a conditioning technique recently developed
by Erwin Bolthausen in the context of spin glass theory.
\end{abstract}
%
%
\section{Introduction and main results}

Given an $n\times N$ matrix $A$, the compressed  sensing reconstruction
problem requires to reconstruct a sparse vector $x_0\in \reals^{N}$ from
a (small) vector of linear observations $y = Ax_0+w\in \reals^n$.
Here $w$ is a noise vector and $A$ is assumed to be known.
 Recently \cite{DMM09} suggested the following first order
\emph{approximate message-passing (AMP)} algorithm for reconstructing $x_0$
given $A,y$. Start with an initial guess $x^0=0$ and proceed by
\begin{align}
x^{t+1}&=\eta_t(A^*z^t+x^t),\label{eq:dmm}\\
z^t &= y - Ax^t+\f{1}{\delta} z^{t-1}
\left\<\eta_{t-1}'(A^*z^{t-1}+x^{t-1})\right\>\, ,\nonumber
\end{align}
for an appropriate sequence of non-linear functions $\{\eta_t\}_{t\ge 0}$.
(Here by convention any variable with negative index is assumed to be $0$.)
 The algorithm succeeds if $x^t$ converges to
a good approximation of $x_0$ (cf. \cite{DMM09} for details).

Throughout this  paper, the matrix $A$ is normalized in such a way that its columns have $\ell_2$ norm\footnote{Recall that the $\ell_p$ norm of a vector
$v$ is $\|v\|_p\equiv(\sum_i|v_i|^p)^{1/p}$.} concentrated around $1$.
Given a vector $x\in \reals^{N}$ and a scalar function
$f:\reals\to\reals$, we write $f(x)$ for the vector
obtained by applying $f$ componentwise. Further, $\delta=n/N$, $\<v\>\equiv
N^{-1}\sum_{i=1}^{N}v_i$ and $A^*$ is the transpose of matrix $A$.

Three findings were presented in \cite{DMM09}:
\begin{itemize}
\item[(1)] For a large class of random matrices $A$, the
behavior of AMP algorithm is accurately described by
a formalism called `state evolution' (SE). Extensive numerical
experiments tested this claim on gaussian, Radamacher, and partial Fourier
matrices;

\item[(2)] The sparsity-undersampling tradeoff of AMP as derived from
SE coincides, for an appropriate choice of the functions $\eta_t$,
with the one of convex optimization approaches. Let us stress that
standard convex optimization algorithms do not scale
to large applications (e.g. to image processing), while the
computational complexity of AMP is as low as the one of the simplest
greedy algorithms;

\item[(3)] As a byproduct of $(1)$ and $(2)$,
SE allows to re-derive reconstruction phase boundaries
earlier determined via random polytope geometry
(see in particular \cite{DoTa05,DoTa08} and references therein).
\end{itemize}
These findings were based on heuristic arguments and
numerical simulations. In this paper we provide the first
rigorous support to finding $(1)$, by proving that SE holds in the large
system limit, for random sensing matrices $A$ with gaussian entries.
Implications on points $(2)$ and $(3)$ will be reported in a forthcoming paper.

Interestingly, state evolution gives access to sharp predictions
that cannot be derived from random polytope geometry.
A prominent example is the noise sensitivity of LASSO,
which is investigated in \cite{NSPT}.

Note that AMP is an approximation to the following message passing algorithm.
For all $i,j\in[N]$ and $a,b\in[n]$ (here and below
$[N]\equiv \{1,2,\ldots,N\}$) start with messages $x_{j\to a}^0=0$
and proceed by
\begin{align}
z_{a\to i}^t &= y_a - \sum_{j\in[N]\bs i}A_{aj}x_{j\to a}^t\, ,
\label{eq:mp1}\\
x_{i\to a}^{t+1}&=\eta_t\left(\sum_{b\in[n]\bs a}A_{bi}z_{b\to i}^t\right).
\nonumber
\end{align}
As argued in \cite{DMM_ITW_I}, AMP accurately approximates message passing in
the large system limit. We refer to appendix \ref{app:AMPderivation}
for an heuristic argument justifying the AMP update rules
(\ref{eq:dmm}) starting from the algorithm (\ref{eq:mp1}).
While this derivation is not necessary for the proofs of this
paper, it can help the reader familiar with message passing
algorithms to
develop the correct intuition.

An important tool for the analysis
of message passing algorithms is provided by density evolution
\cite{RiU08}. Density evolution is known to hold asymptotically
for sequences of sparse graphs that are locally tree-like.
The factor graph underlying the algorithm (\ref{eq:mp1}) is dense:
indeed it is the complete bipartite graph.
State evolution can be regarded (in a very precise sense) as the analogue
of density evolution for dense graphs.

For the sake of concreteness, we will focus in this Section
on the algorithm (\ref{eq:dmm}), and will keep to the
compressed sensing language.
Nevertheless our analysis applies to a much larger family
of message passing algorithms on dense graphs, for instance the
multi-user detection algorithm studied in
\cite{Kab03,NeirottiSaad,MoT06}.  Applications to
such algorithms are discussed in Section \ref{sec:Examples}.
Section \ref{sec:Analysis}  describes an even more general formulation,
as well as the proof of our theorems. Finally, Section
\ref{sec:Symmetric} describes a generalization to the case of symmetric
matrices $A$ that is directly related to the work of Erwin Bolthausen \cite{Bolthausen} .

It is important to mention that the algorithms \eqref{eq:dmm} and
\eqref{eq:mp1}
are completely different from gaussian belief propagation (BP).
The gaussian assumption refers indeed to the distribution of the matrix
entries, not to the variables to be inferred.
More generally, none of the existing rigorous results for BP seem to be applicable here.

It is  remarkable that density evolution (in its special incarnation,
SE) holds for dense graphs. This is at odds with the standard argument
used for justifying density evolution so far:
`density evolution works \emph{because} the graph is locally tree-like.'
To the best of our knowledge, the approach developed here
is the first one that overcomes the limitations of the standard argument
(a discussion of earlier literature is provided in Section \ref{sec:Related}).
%
%
\subsection{Main result}\label{subsec:main-result}

We begin with some missing definitions for algorithm (\ref{eq:dmm}). We assume
\begin{eqnarray}
y = Ax_0+w\, ,\label{eq:Model}
\end{eqnarray}
with $w\in \reals^n$ a vector with i.i.d. entries with mean
$0$ and variance $\sigma^2$.
In Section \ref{sec:General}, we will show that the
i.i.d. assumption can be relaxed to existence of a weak limit for the empirical distribution of $w$ with certain moment conditions.
Further, let $\{\eta_t\}_{t\ge 0}$ be a sequence of scalar functions
$\eta_t:\reals\to\reals$ which we assume to be Lipschitz continuous
(and hence almost everywhere differentiable).
Define the sequence of vectors $\{x^t\}_{t\ge 0}$,
$x^{t}\in\reals^N$,  $\{z^t\}_{t\ge 0}$,
$z^{t}\in\reals^n$, through Eqs.~(\ref{eq:dmm}).

Next, let us define formally state evolution.
Given a probability distribution $p_{X_0}$, let $\tau_0^2\equiv
\sigma^2+\E\{X_0^2\}/\delta$, and define recursively for $t\ge 0$,
\begin{align}
\tau_{t+1}^2&  =\sigma^2+\frac{1}{\delta}
\E\bigg\{[\eta_t(X_0+\tau_tZ)-X_0]^2\bigg\}\, ,\label{eq:SERecursion}
\end{align}
with $X_0\sim p_{X_0}$ and $Z\sim \normal(0,1)$ independent
from $X_0$. We will use the term \emph{state evolution} to refer
both to the recursion (\ref{eq:SERecursion})
(or its more general version introduced in Section \ref{sec:General})
and to the
sequence $\{\tau_t\}_{t\ge 0}$ that it defines.

Let us denote the empirical
distribution\footnote{The probability distribution that puts a point mass $1/N$ at each of the $N$
entries of the vector.} of a  vector $x_0\in\reals^N$ by $\empr_{x_0}$.
Further, for $k\ge1$ we say a function $\phi:\reals^m\to\reals$ is \emph{pseudo-Lipschitz}
of order $k$ and denote it by $\phi\in \psL(k)$ if there exists a constant
$L>0$ such that, for all $x,y\in\reals^m$:
\begin{eqnarray}
|\phi(x)-\phi(y)|\le L(1+\|x\|^{k-1}+\|y\|^{k-1})\, \|x-y\|\, .
\end{eqnarray}
Notice that when $\phi\in \psL(k)$, the following two properties follow:
\begin{itemize}
\item[$(i)$] There is a constant $L'$ such that for all $x\in\reals^m$: $|\phi(x)|\leq L'(1+\|x\|^k)$.
\item[$(ii)$] $\phi$ is locally Lipschitz, that is for any $M>0$ there exist
a constant $L_{M,m}<\infty$ such that for all $x,y\in [-M, M]^m$,
\[
|\phi(x)-\phi(y)|\leq  L_{M,m}\|x-y\|.
\]
Further, $L_{M,m}\leq c[1+ (M\sqrt{m})^{k-1}]$ for some constant $c$.
\end{itemize}
In the following we shall use generically $L$ for Lipschitz constants
entering bounds of this type. It is understood (and will not
be mentioned explicitly) that the constant must be properly adjusted at
various passages.

\begin{theorem}\label{thm:main}
Let $\{A(N)\}_{N\ge 0}$ be a sequence of sensing
matrices $A\in\reals^{n\times N}$
indexed by $N$, with i.i.d. entries $A_{ij}\sim \normal(0,1/n)$, and assume $n/N\to\delta\in (0,\infty)$.
Consider further a sequence of signals $\{x_0(N)\}_{N\ge 0}$, whose
empirical distributions converge weakly to a probability measure
$p_{X_0}$ on $\reals$  with bounded $(2k-2)^{th}$ moment, and
assume $\E_{\empr_{x_0(N)}}(X_0^{2k-2})\to\E_{p_{X_0}}(X_0^{2k-2})$ as $N\to\infty$ for some $k\ge 2$.
Also, assume the noise $w$ has iid entries with a distribution $p_W$ that has bounded $(2k-2)^{th}$ moment.
Then, for any pseudo-Lipschitz function $\psi:\reals^2\to\reals$
of order $k$ and all $t\ge 0$, almost surely
\begin{eqnarray}
\lim_{N\to\infty}\f{1}{N}\sum_{i=1}^N\psi(x_i^{t+1},x_{0,i})
=\E\Big[ \psi\big(\eta_{t}(X_0+\tau_{t} Z),X_0\big)\Big]\,,
\label{eq:SE2}
\end{eqnarray}
with $X_0\sim p_{X_0}$ and $Z\sim \normal(0,1)$ independent.
\end{theorem}
Up to a trivial change of variables, this is a formalization of
the findings of \cite{DMM09} (cf. in particular Eqs. (7), (8) and Finding 2
in that paper).

As an immediate consequence of the above theorem we have the following
\emph{decoupling principle} implying that a typical (finite) subset
of the coordinates of $x^t$ are asymptotically independent.
\begin{corollary}[Decoupling principle]\label{coro:Decoupling}
Under the assumption of Theorem \ref{thm:main}, fix  $\ell\ge 2$,
let $\psi:\reals^{2\ell}\to\reals$ be any Lipschitz function,
and denote by ${\sf E}$ expectation with respect to a
uniformly random subset of distinct indices
$J(1),\dots,J(\ell)\in [N]$.

Then for all $t>0$, almost surely
\begin{eqnarray}
\lim_{N\to\infty}{\sf E}\, \psi(x_{J(1)}^{t},\dots,x_{J(\ell)}^t,
x_{0,J(1)},\dots,x_{0,J(\ell)})
=\E\Big\{\psi\big(
\widehat{X}_{1},\dots,\widehat{X}_{\ell},X_{0,1},\dots,
X_{0,\ell}\big)\Big\}\, ,\label{eq:Kconv}
\end{eqnarray}
where
$\widehat{X}_i\equiv\eta_{t-1}(X_{0,i}+\tau_{t-1} Z_i)$
for  $X_{0,i}\sim p_{X_0}$ and $Z_i\sim \normal(0,1)$,
$i=1,\dots,\ell$ mutually independent.
\end{corollary}
For the proof of this corollary we refer to Section
\ref{sec:CoroProof}.
%
%
\subsection{Universality}

Our proof technique heavily relies on the assumption that $A(N)$
is gaussian. Nevertheless, we expect the convergence expressed
in Theorem \ref{thm:main} to be a fairly general result.
In particular, we expect it to hold for matrices with i.i.d.
entries with zero mean and variance $1/n$, under a suitable
moment condition. This type of \emph{universality} is quite common
in random matrix theory, and several arguments suggest that it
should hold in the present case. For instance,
it is possible to prove that state evolution holds for
this broader class of random matrices when $\eta_t(\,\cdot\,)$
is affine. Also, the heuristic argument discussed in the next
section is clearly insensitive to the details of distribution of the
entries.

Numerical evidence presented in \cite{DMM09}
(we refer in particular to the online supplement)
suggests that state evolution might hold for an even broader class of
matrices. Determining the domain of such an universality class is an
outstanding open problem.
%
%
\subsection{State evolution: the basic intuition}
\label{sec:Intuition}

The state evolution recursion has a simple heuristic
description, that is useful to present here since it
clarifies the difficulties  involved in the proof.
In particular, this description brings up the key role
played by the last term in the update equation
for $z^t$, that we will call the `Onsager term',
following \cite{DMM09}.

Consider again the recursion (\ref{eq:dmm}), but
introduce the following three modifications:
$(i)$ Replace the random matrix $A$ with a new
independent copy $A(t)$ at each iteration $t$;
$(ii)$ Correspondingly replace the observation vector $y$
with $y^t=A(t)x_0+w$;
$(iii)$ Eliminate the last term in the update equation for $z^t$.
We thus get the following dynamics:
\begin{align}
x^{t+1}& = \eta_t(A(t)^*z^t+x^t)\, ,\\
z^{t}& = y^t-A(t)x^t\, ,
\end{align}
where $A(0),A(1),A(2),\dots$ are i.i.d. matrices of dimensions $n\times N$
with i.i.d. entries $A_{ij}(t)\sim \normal(0,1/n)$.
(Notice that, unlike in the rest of the paper, we use here the
argument of $A$ to denote the iteration number, and not the matrix
dimensions.)

This recursion is most
conveniently written by eliminating $z^t$:
\begin{align}
x^{t+1}& = \eta_t\big(A(t)^*y^t+(\identity-A(t)^*A(t))x^t\big)\, ,\nonumber\\
       & = \eta_t\big(x_0+A(t)^*w+B(t)(x^t-x_0)\big)\, ,
\label{eq:ModifiedRecursionNoAlg}
\end{align}
where we defined $B(t) = \identity-A(t)^{*}A(t)\in\reals^{N\times N}$.
Notice that this recursion does not correspond to any concrete algorithm,
since the matrix $A$ changes from iteration to iteration.
It is nevertheless useful for developing intuition.

Using the central limit theorem, it is easy
to show that each entry of $B(t)$ is approximately
normal, with  zero mean and variance $1/n$.
Further, distinct entries are approximately pairwise independent.
Therefore, if we let
$\htau_t^2 = \lim_{N\to\infty}\|x^t-x_0\|^2/N$, we obtain that $B(t)(x^t-x_0)$
converges to a vector with i.i.d. normal entries with $0$ mean and
variance $N\htau_t^2/n = \htau_t^2/\delta$. Notice that this is true
because $A(t)$ is independent of $\{A(s)\}_{1\le s\le t-1}$ and,
in particular, of $(x^t-x^0)$.

Conditional on $w$, $A(t)^*w$ is a vector of i.i.d. normal
entries with mean $0$  and variance $(1/n)\|w\|^2$
which converges by the law of large numbers to $\sigma^2$.
A slightly longer exercise shows that these entries are
approximately independent from the ones of $B(t)(x^t-x_0)$.
Summarizing, each entry of the vector in the argument
of $\eta_t$ in Eq.~(\ref{eq:ModifiedRecursionNoAlg})
converges to $X_0+\tau_t Z$ with $Z\sim\normal(0,1)$
independent of $X_0$, and
\begin{align}
\tau_t^2 &= \sigma^2+\frac{1}{\delta}\htau_t^2\, ,\label{eq:FirstHeuristic}\\
\htau_t^2 & =\lim_{N\to\infty}\frac{1}{N}\|x^t-x_0\|^2\, .\nonumber
\end{align}
On the other hand, by Eq.~(\ref{eq:ModifiedRecursionNoAlg}),
each entry of $x^{t+1}-x_0$ converges to $\eta_t(X_0+\tau_t\, Z)-X_0$,
and therefore
\begin{align}
\htau_{t+1}^2 & =\lim_{N\to\infty}\frac{1}{N}\|x^{t+1}-x_0\|^2
=\E\big\{[\eta_t(X_0+\tau_t\, Z)-X_0]^2\big\}\, .\label{eq:SecondHeuristic}
\end{align}
Using together Eq.~(\ref{eq:FirstHeuristic}) and (\ref{eq:SecondHeuristic})
we finally obtain the state evolution recursion, Eq.~(\ref{eq:SERecursion}).

We conclude that state evolution would hold if the matrix $A$
was drawn independently from the same gaussian distribution at
each iteration. In the case of interest, $A$ does not change across
iterations, and the above argument falls apart because
$x^t$ and $A$ are dependent. This dependency is non-negligible even
in the large system limit $N\to\infty$. This point can be clarified by
considering the iteration
\begin{align}
x^{t+1}& = \eta_t(A^*z^t+x^t)\, ,\label{IST}\\
z^{t}& = y-Ax^t\, ,\nonumber
\end{align}
with a matrix $A$ constant across iterations.
This iteration is the basis of several algorithms in
compressed sensing, most notably the so-called `iterative soft thresholding'
\cite{DaDeDe04}. Such algorithms have been the object of
great interest because of the high computational cost of
standard convex optimization methods in large scale applications.

Numerical studies of iterative soft thresholding
\cite{DonohoMalekiTuned,DMM09} show that its behavior is dramatically
different from the one in Eq.~(\ref{eq:dmm}) and in particular
\emph{state evolution does not hold for the
iterative soft thresholding iteration} (\ref{IST}), even in the large
system limit.

This is not a surprise: the correlations between $A$ and $x^t$
simply cannot be neglected.
On the other hand, adding the Onsager term leads to an asymptotic cancelation
of these correlations. As a consequence, state evolution
holds for the AMP iteration (\ref{eq:dmm}) despite the fact that the matrix is kept
constant.
%
%
\subsection{Related literature}
\label{sec:Related}

As mentioned, the standard argument for justifying
density evolution relies on the locally-tree like
structure of the underlying graph. This argument was
developed and systematically exploited for the analysis of low-density
parity-check (LDPC) codes under iterative decoding \cite{RiU08}.
In this context,
density evolution provides an exact tool for computing asymptotic
thresholds of code ensembles based on sparse graph constructions.
Optimization of these thresholds has been a major design principle
in LDPC codes.

The locally tree-like property is a special case of local
weak convergence.
Local weak convergence of graph sequences was first defined
and studied in probability theory by Benjamini and Schramm
\cite{BenjaminiSchramm}, and then greatly developed by David Aldous
\cite{AldousSteele}, in particular to study the so called
`random assignment problem' \cite{AldousAssignment}.
Loosely  speaking, local weak convergence allows to treat sequences of
graphs of increasing size, such that the neighborhood of a node
converges to a well defined limit object.

The random assignment problem is defined as a distribution of random
instances of the assignment problem on
complete bipartite graphs. In particular, such graphs are
not locally tree-like.
Nevertheless, they admit a rather simple local weak limit (called
the PWIT), which is a tree. The basic reason is that only a sparse subgraph
of the complete bipartite graph is relevant for the minimum cost assignment,
namely the one of edges with small cost.
One concrete way to derive density evolution in this case is indeed
to eliminate all the edges of cost larger than --say--
$\Delta_n/n$ with $\Delta_n$ diverging slowly with the graph size $n$.
 The resulting graph is sparse and one can apply standard arguments
(cf. \cite{MezardMontanari} for an outline of this argument).
A more sophisticated argument was presented in
\cite{SalezShah} which nevertheless uses
the existence of a non-trivial local weak
limit, and the fact that only a sparse subgraph is relevant (Lemma 4.1
in \cite{SalezShah}).

This reduction to a sparse graph, and hence to a limit tree,
is impossible in the class of algorithms
studied in our paper: the
algorithm iteration cannot be approximated
by an iteration on a sparse graph (at least not on an
instance-by-instance basis).
This corresponds to the fact that no (simple) local weak limit
exists in our case. The underlying graph is the complete
bipartite graph with vertex sets $[N]\equiv\{1,2,\ldots,N\}$ and $[n]\equiv\{1,2,\ldots,n\}$,
and edge-weights $A_{ai}$ for all $(a,i)\in [n]\times [N]$.
If we choose a node $i\in [N]$ as root, its depth-$1$ neighborhood
consists of $[n]$ node, each carrying a weight of order $1/\sqrt{n}$.
Even this small neighborhood has no simple local weak limit.

 This difference is analogous to the difference between
mean-field spin glasses (e.g. the Sherrington-Kirkpatrick model)
and the random assignment problem \cite[Chapter 7]{TalagrandBookLatest}.
As a consequence, our proof does not rely on local weak convergence,
and has to deal directly with the intricacies of graphs with many short cycles.

The theorem proved in this paper is not only relevant
for \cite{DMM09} but for a larger context as well. First of all,
following the work by Tanaka \cite{TanakaCDMA},
hundreds of papers have been published
in information theory using the replica method to study multi-user
detection problems.
In its replica-symmetric version, the replica method typically predicts
the system performances through the solution of
a system of non-linear equations, which coincide
with the fixed point equations for state evolution.
The present result provides a rigorous
foundation to that line of work, along with the analysis of a concrete
algorithm that achieves those performance.
Further, \cite{GuoVerdu} insisted on the role of a `decoupling principle'
that emerges from the replica method, and on the insight it provides.
Corollary \ref{coro:Decoupling} indeed proves a specific form of
this decoupling principle.

A more recent line of works uses the replica method to
study typical performances of compressed sensing methods.
Although non-rigorous and limited to asymptotic statements,
the replica method has the advantage of providing sharp predictions.
Standard techniques instead predict performances up to undetermined
multiplicative constants. The determination of these constants
can be of guidance for practical applications. This motivated
several groups to publish results based on the replica method
\cite{Goyal,KabashimaTanaka,BaronGuoShamai}.
The present paper provides a rigorous foundation to this work as well.

%
%
\section{Examples}
\label{sec:Examples}

In this section we discuss in greater detail
some of the applications of Theorem \ref{thm:main} to specific problems.
To be definite, it is convenient to keep in mind
a specific observable for applying  Theorem \ref{thm:main}.
If we choose the test function $\psi(x,y) = (x-y)^2$, we get almost surely
\begin{eqnarray}
\lim_{N\to\infty}\frac{1}{N}\|x^t-x_0\|^2 = (\tau_t^2-\sigma^2)\delta\, .
\label{eq:MSE}
\end{eqnarray}
Therefore state evolution allows to predict the mean square error of the
iterative algorithm (\ref{eq:dmm}). More generally, state evolution can be
used to estimate $\ell_p$ distances for $p\leq k$ through
\begin{eqnarray}
\lim_{N\to\infty}\frac{1}{N}\|x^t-x_0\|^p_p =
\E\big\{[\eta_{t-1}(X_0+\tau_{t-1}Z)-X_0]^p\big\}\,,
\label{eq:ellp}
\end{eqnarray}
almost surely.
%
%
\subsection{Linear estimation}

As a warm-up example consider the case in which the
\emph{a priori} distribution of $x_0$ is gaussian, namely
its entries are i.i.d. $\normal(0,v^2)$.
It is a consequence of state evolution that the
optimal AMP algorithm makes use of linear scalar estimators
\begin{eqnarray}
\eta_t(x) = \lambda_t\, x\, .
\end{eqnarray}
Clearly, such functions are Lipschitz continuous, for any
$\lambda_t$ finite. The AMP algorithm (\ref{eq:dmm}) becomes
\begin{align}
x^{t+1}&=\lambda_t(A^*z^t+x^t),\label{eq:linear}\\
z^t &= y - Ax^t+(\lambda_{t-1}/\delta) \, z^{t-1}
\, .\nonumber
\end{align}
State evolution reads
\begin{align}
\tau_{t+1}^2&  =\sigma^2+\frac{1}{\delta}
(1-\lambda_t)^2v^2+\frac{1}{\delta}\lambda_t^2\tau_t^2\, .\label{LinearSE}
\end{align}
Theorem \ref{thm:main} also shows that
the empirical distribution of  $\{(A^*z^t+x^t)_i-x_{0,i}\}_{i\in [N]}$
is asymptotically gaussian with
mean $0$ and variance $\tau_t^2$. Hence, the optimal choice of $\lambda_t$
is
\begin{eqnarray}
\lambda_t = \frac{v^2}{v^2+\tau_t^2}\, .
\end{eqnarray}
Notice that this also minimizes the right hand side of
Eq.~(\ref{LinearSE}). Under this choice, the recursion
(\ref{LinearSE}) yields
\begin{eqnarray}
\tau_{t+1}^2 = \sigma^2 + \frac{1}{\delta}\, \frac{v^2\tau_t^2}
{v^2+\tau_t^2}\, .
\end{eqnarray}
The right hand side is a concave function of $\tau_t^2$,
and is easy to show that $\tau_t\to\tau_{\infty}$ exponentially fast,
where, for $c= (1-\delta)/\delta$,
\begin{eqnarray}
\tau_{\infty}^2 = \frac{1}{2}\Big\{
(\sigma^2+cv^2)+\sqrt{(\sigma^2+cv^2)^2+4\sigma^2v^2}\Big\}
\, .
\end{eqnarray}
The mean square error of the resulting algorithm
is estimated via Eq.~(\ref{eq:MSE}). In particular, under
the optimal choice of $\lambda_t$, the latter converges to
$(\tau_{\infty}^2-\sigma^2)\delta$ with $\tau_{\infty}$
given as above, thus yielding
\begin{eqnarray}
\lim_{t\to\infty}\lim_{N\to\infty}\frac{1}{N}\|x^t-x_0\|^2 = \frac{\delta}{2}\Big\{
(-\sigma^2+cv^2)+\sqrt{(\sigma^2+cv^2)^2+4\sigma^2v^2}\Big\}
\, .\label{eq:linearMSE}
\end{eqnarray}

We recall that the  asymptotic mean square error of optimal
(MMSE) linear estimation was computed by Tse-Hanly and Verd\'{u}-Shamai
in the case of random matrices $A$ with i.i.d. entries
\cite{TseHanly,VerduShamaiCDMA}. The motivation came from
the analysis of multiuser receivers.
The resulting MSE coincides with the value predicted
in Eq.~(\ref{eq:linearMSE}), thus showing that --in the linear case--
the AMP algorithm is asymptotically equivalent to the MMSE
estimator.

Notice that the computation of the MMSE in
\cite{TseHanly,VerduShamaiCDMA} relied heavily on the Marcenko-Pastur
law for the limit spectral  law of Wishart matrices
\cite{MarcenkoPastur}. Conversely, any calculation of the MMSE
as a function of the noise variance $\sigma^2$ gives access to
the asymptotic Stieltjis transform of the spectral measure of $A$.
This suggests that state evolution is a  non-trivial result already in
the case of linear $\eta_t(\,\cdot\,)$, since it can be used to
derive the Marcenko-Pastur law in random matrix theory.

%
%
\subsection{Compressed sensing via soft thresholding}

In this case the vector $x_0$ is $\ell$ sparse (i.e. it has at most
$\ell$ non-vanishing entries). Assuming that the empirical distribution
of $x_0$ converges to the probability measure $p_{X_0}$, it
is also natural to assume $\ell/N\to\ve$ as $N\to\infty$ with
\begin{eqnarray}
\prob\{X_0\neq  0\} = \ve\, . \label{eq:SparseSource}
\end{eqnarray}
(Indeed  Theorem \ref{thm:main} accommodates for a more general
behavior, since $\empr_{x_0(N)}$ is only required to converge
weakly.)

In \cite{DMM09}, the authors proposed an algorithm of the form
(\ref{eq:dmm}) with $\eta_t(x) = \eta(x;\theta_t)$ a sequence of
soft-threshold functions
\begin{eqnarray}
\eta(x;\theta) = \left\{
\begin{array}{ll}
(x-\theta) & \mbox{ if $x>\theta$,}\\
0 & \mbox{ if $-\theta\le x\le \theta$,}\\
(x+\theta) & \mbox{ if $x<-\theta$.}
\end{array}\right.\label{eq:SoftThreshold}
\end{eqnarray}
The function $x\mapsto \eta(x;\theta)$ is non-linear but nevertheless
it is Lipschitz continuous. Therefore
Theorem \ref{thm:main} applies to this case, and allows
to predict the asymptotic mean square error using
Eqs.~(\ref{eq:SERecursion}) and (\ref{eq:SE2}).

This choice of the nonlinearity $\eta_t$
is close to the optimal in minimax sense. Indeed,
a substantial literature
(see e.g. \cite{DJ94a,DJ94b}) studies the problem of estimating
the scalar $X_0$ from the noisy observation
\begin{eqnarray}
Y=X_0+Z\, ,\label{eq:ScalarProb}
\end{eqnarray}
with $Z\sim\normal(0,s^2)$. For an appropriate choice
of the threshold $\theta=\theta(\ve,s)$, and $\ve\downarrow 0$
(very sparse sources), the  soft thresholding estimator
was proved to be minimax optimal, i.e. to achieve the minimum
worst-case MSE over the class (\ref{eq:SparseSource}).
State evolution allows to deduce that
the choice (\ref{eq:SoftThreshold}) yields the best algorithm
of the form (\ref{eq:dmm}) for estimating sparse vectors, over the
worst-case  vector $x_0$ \cite{DMM09}.

It is argued in \cite{DMM09,NSPT}, and proved in \cite{InPreparation}
in the case of gaussian matrices,
that the asymptotic MSE of AMP coincides with the one of a popular
convex optimization estimation technique, known as the LASSO.
The above argument is suggestive of a possible way to prove minimax
optimality of the LASSO.

Finally, state evolution provides a systematic way of
improving the choice of the non-linearities $\eta_t$
when the class of signal changes. The basic idea is to choose
the function $\eta_t$ that minimizes the right-hand side of
Eq.~(\ref{eq:SERecursion}) in minimax sense.
This corresponds to constructing minimax MMSE estimators
for the scalar problem (\ref{eq:ScalarProb}).
For instance,
the limit case in which the distribution of $X_0$ is known,
the MMSE estimator is simply conditional expectation,
which leads to the choice
\begin{eqnarray}
\eta_t(x) = \E\{X_0\, |\, X_0+\tau_t\, Z=x\}\, ,\label{eq:GeneralMMSE}
\end{eqnarray}
with $Z\sim\normal(0,1)$. In other words, the very choice
of the non-linearities is dictated by the gaussian
convergence phenomenon described in Theorem \ref{thm:main}.
%
%
\subsection{Multi-User Detection}

The model (\ref{eq:Model}) is used to describe the
input-output relation in code division multiple access (CDMA)
channel. The matrix $A$ contains the users' signatures.
A frequently used setting for theoretical analysis
consists in taking the large system limit with $n/N\to\delta$
giving the spreading factor,
and in assuming that the signatures (and hence $A$) have i.i.d. components.
The entries $x_{0,i}$ belong to
the signal constellation used by the system.
For the sake of
simplicity, we consider the case of antipodal signaling,
i.e. $x_{0,i}\in\{+1,-1\}$ uniformly at random.
Other signal constellations can also be treated applying
our Theorem \ref{thm:main}. The hypothesis
that $x_{0}$ is independent from $A$ is also standard
in this context and justified by the remark that the transmitted
information is independent from the signatures. Further,
the source-channel separation theorem naturally leads to the uniform
distribution.

Following \cite{Kab03,NeirottiSaad,MoT06}
we take
\begin{eqnarray}
\eta_t(x) = \tanh\big\{x/\tau_t^2\big\}\, .
\end{eqnarray}
The rationale for this choice is that it gives the conditional expectation
of a uniformly random signal $X_{0}\in\{+1,-1\}$, given the observation
$X_{0}+\tau_t\, Z=x$\, for\, $Z\sim\normal(0,1)$ gaussian noise.
This is therefore a special case of the rule (\ref{eq:GeneralMMSE})
and by the argument given there, it achieves minimal mean-square error
within the class of algorithms (\ref{eq:dmm}).

The algorithm (\ref{eq:dmm}) reads in this case
\begin{align}
x^{t+1}&=\tanh\Big\{\frac{1}{\tau_t^2}(A^*z^t+x^t)\Big\}
\, ,\label{eq:CDMA}\\
z^t &= y - Ax^t+\f{z^{t-1}}{\delta\tau_t^2} \Big\{1-
\left\<\tanh^2\big[(A^*z^t+x^t)/\tau_t^2]\right\>\Big\}
\, .\nonumber
\end{align}

State evolution yields
\begin{eqnarray}
\tau_{t+1}^2 = \sigma^2+\frac{1}{\delta}\, \E\Big\{
\big[\tanh\big(\tau_t^{-2}+\tau_t^{-1}Z\big)-1\big]^2\Big\}\, .
\label{eq:SE_Multiuser}
\end{eqnarray}
This state evolution recursion was proved in \cite{MoT06}
for properly chosen sparse signature matrices $A$.
Theorem \ref{thm:main} provides the first generalization to
the more relevant case of dense signatures.

As mentioned in Section \ref{sec:Related}, Tanaka used the replica
method to compute the asymptotic performance of a MMSE receiver.
The expressions obtained through this method correspond to
a fixed point of the recursion (\ref{eq:SE_Multiuser}).
It was further proved in \cite{MoT06} that, whenever the fixed point is unique,
this prediction is asymptotically correct. For such values
of the parameters, we deduce that the AMP algorithm is
asymptotically equivalent to the MMSE receiver.

Let us point out that, in a practical setting,
it might be inconvenient to estimate the noise variance and/or
to change the function $\eta_t$ across iterations.
Several authors (see for instance  \cite{Tanaka0}) used the function
\begin{eqnarray}
\eta_t(x) = \tanh\big\{\beta x\big\}\, .
\end{eqnarray}
State evolution can be applied in this case as well (for any
finite $\beta$) and reads
\begin{eqnarray}
\tau_{t+1}^2 = \sigma^2+\frac{1}{\delta}\, \E\Big\{
\big[\tanh\big(\beta+\beta\tau_tZ\big)-1\big]^2\Big\}\, .
\end{eqnarray}
On the other hand the case $\beta\to\infty$ is not
covered by our Theorem \ref{thm:main}, since it corresponds to the
discontinuous function $\eta_t(x) = {\rm sign}(x)$.
%
%
\section{Proof}
\label{sec:Analysis}

The proof is based on a conditioning technique developed by
Erwin Bolthausen for the analysis of the so-called TAP
equations in spin glass theory \cite{Bolthausen}.
Related ideas can also be found in \cite{DonohoIterative}.

In the next section, we provide a high-level description of the
conditioning technique, by using a simpler
type of recursion as reference. We will then
introduce some new notations and state and prove a more
general result than Theorem \ref{thm:main}.

%
%
\subsection{The conditioning technique: an informal description}
\label{sec:HighLevel}

For understanding the conditioning technique, it is convenient to
consider a somewhat simpler setting, namely the one of symmetric
matrices. This will be discussed more formally in
Section \ref{sec:Symmetric}. Let $G=A^*+A$ where
$A\in\reals^{N\times N}$ has i.i.d. entries $A_{ij}\sim\normal(0,(2N)^{-1})$.
We consider the iteration
\begin{align}
h^{t+1}&=Gm^{t}-\lambda_t m^{t-1},\label{eq:TAP-1st-appearence}\\
m^t&=f(h^t)\, ,
\end{align}
for $f:\reals\to\reals$
a non-linear function and $m^{-1} = 0$. For the sake of simplicity,
$h^0=0$. The correct expression for the scalar
$\lambda_t$ is provided in Section \ref{sec:Symmetric},
and state evolution holds only if this value is used.
On the other hand, this expression is not important for our
informal discussion here.

Consider the first iteration. By definition $m^{-1}=0$,
whence $h^1=Gf(0)$ is a vector with i.i.d. gaussian components
with variance $\|f(0)\|^2/N$. This follows in particular by the rotational
invariance of the distribution of $G$, which implies that, for a deterministic
vector $v$, $Gv$ is distributed as
$\|v\|Ge_1$ for $e_1$ the first vector of the standard basis of $\reals^N$
(see also Lemma \ref{lem:prop-Gaussian-matrix} below).

Now consider the $t^{\rm th}$ iteration (i.e., $h^{t+1}=Gm^t-\lambda_t m^{t-1}$).
The problem in repeating the above argument is
that $G$ and $f(m^t)$ are dependent. For instance $f(m^t)$
might \emph{a priori} align with the minimum eigenvector of $G$.
More generally the problem is that
$G$ is not independent from the $\sigma$-algebra
$\sigal{S}_t$ generated by $\{h^0,h^1,\dots,h^t\}$.

The key idea in the conditioning technique is to avoid
computing the conditional distribution of $m^t$ given $G$.
We  instead compute the \emph{conditional distribution of $G$ given
$\sigal{S}_t$}.

The next important remark is that
$\sigal{S}_t$ contains  $\{m^0,m^1,\dots,m^t\}$
as well. Conditioning on $\sigal{S}_t$ is therefore equivalent to
conditioning on the event
\begin{eqnarray*}
{\cal E}_t \equiv \{h^1 +\lambda^{0}m^{-1}= Gm^{0},\dots
,h^t +\lambda^{t-1}m^{t-2}= Gm^{t-1}\}\,.
\end{eqnarray*}
which is in turn equivalent to making a set of linear observations
of $G$.

At this point, the assumption that $G$ is gaussian plays a crucial role.
The conditional distribution of a gaussian random variable $G$
given linear observations is the same as its conditional
expectation plus the projection of an independent gaussian.
In formulae:
\begin{eqnarray*}
G|_{\sigal{S}_t} &\deq& G|_{{\cal E}_t} \deq \E\{G|\sigal{S}_t\}
+ {\rm P}^t_{\perp} G^{\rm new}{\rm P}^t_{\perp}\, ,
\end{eqnarray*}
with ${\rm P}^t_{\perp}$ an appropriate projector. If we
write $E_t \equiv \E\{G|\sigal{S}_t\}$, we have
\begin{eqnarray*}
Gm^t|_{\sigal{S}_t} &\deq&  G^{\rm new}({\rm P}^t_{\perp}m^t)-(I-
{\rm P}^t_{\perp}) G^{\rm new}({\rm P}^t_{\perp}m^t)+E_tm^t\, .
\end{eqnarray*}
We refer to the actual proof for a calculation of the various terms
involved.

Each of the above terms can be written
explicitly as a function of the observed values $\{m^0,m^1,\dots,m^t\}$
and of the new gaussian random variables $G^{\rm new}$.
The first term $G^{\rm new}({\rm P}^t_{\perp}m^t)$ is clearly
gaussian. The other terms are not.
 In order to control them, we will proceed by induction over $t$
and use an appropriate strong law of large numbers for triangular
arrays. The key phenomenon
is that the only non-gaussian term that does not vanish in the
large system limit cancels with the term $-\lambda_tm^{t-1}$
in recursion \eqref{eq:TAP-1st-appearence}, thus implying the claimed gaussianity of $h^{t+1}$.
%
%
\subsection{A general result}
\label{sec:General}

We describe now a more general recursion than in Eq.~(\ref{eq:dmm}).
In the next section we show that the AMP algorithm (\ref{eq:dmm})
can be regarded as a special case of the recursion defined here.


The algorithm is defined by two sequences of functions
$\{f_t\}_{t\ge 0}$, $\{g_t\}_{t\ge 0}$, where for each $t\ge 0$,
$f_t:\reals^2\to\reals$ and $g_t:\reals^2\to\reals$
are assumed to be Lipschitz continuous. Recall that
Lipschitz functions are continuous,
and are almost everywhere continuously differentiable with a bounded
derivative.
As before, given $a,b\in \reals^{K}$,
we write $f_t(a,b)$ for the vector obtained by applying componentwise
$f_t$ to $a$, $b$. When $b$ is clear from the context we will just write, with an abuse of notation,
$f_t(a)$.
We will use analogous notations for $g_t$.

Given $w\in\reals^n$ and $x_0\in\reals^N$,
define the sequence of vectors
$h^t,q^t\in\reals^N$ and
$z^t,m^t\in\reals^n$, by fixing initial condition $q^0$, and
obtaining $\{b^t\}_{t\ge0}$, $\{m^t\}_{t\ge0}$,
$\{h^t\}_{t\ge1}$, and $\{q^t\}_{t\ge1}$ through
\begin{eqnarray}
h^{t+1}&=& A^*m^t-\xi_t\, q^t\, ,\;\;\;\;\;\;\;\;
m^t=g_t(b^t,w)\, ,\nonumber\\
b^t&=& A\,q^t-\lambda_t m^{t-1} \, ,\;\;\;\;\;\;\;
q^t=f_t(h^t,x_0)\, ,\label{eq:mpMain}
\end{eqnarray}
where $\xi_t = \< g'_t(b^t,w)\>$, $\lambda_t=\f{1}{\delta}
\< f'_t(h^{t},x^0)\>$ (both derivatives are with respect to the first
argument), and we recall that --by definition-- $m^{-1}=0$.

Assume that the limit
\begin{eqnarray}
\sigma^2_{0}\equiv\lim_{N\to\infty}\frac{1}{N\delta}\|q^0\|^2\label{eq:limit-<||q0||^2>}
\end{eqnarray}
exists, is positive and finite, for a sequence of initial conditions of increasing dimensions.
State evolution defines quantities $\{\tau_t^2\}_{t\ge0}$ and $\{\sigma_t^2\}_{t\ge0}$ via
\begin{eqnarray}
\tau_{t}^2 = \E\big\{ g_t(\sigma_{t} Z,W)^2\big\}\, ,
\;\;\;\;\; \sigma_{t}^2 = \frac{1}{\delta}\,
\E\big\{ f_t(\tau_{t-1} Z,X_0)^2\big\}\, ,\label{eq:tau-sigma-recursion}
\end{eqnarray}
where $W\sim \pW$ and $X_0\sim p_{X_0}$ are independent of $Z\sim\normal(0,1)$.
Further, recall the notion of pseudo-Lipschitz function for $k>1$ from Section \ref{subsec:main-result}.
We have the following general result.
\begin{theorem}\label{thm:general}
Let $\{q_0(N)\}_{N\ge 0}$ and $\{A(N)\}_{N\ge 0}$ be, respectively,
a sequence of initial conditions and a sequence of matrices $A\in\reals^{n\times N}$
indexed by $N$ with i.i.d. entries $A_{ij}\sim \normal(0,1/n)$.
Assume $n/N\to\delta\in (0,\infty)$. Consider sequences of vectors $\{x_0(N),w(N)\}_{N\ge 0}$, whose
empirical distributions converge weakly to  probability measures
$p_{X_0}$ and $\pW$ on $\reals$  with bounded $(2k-2)^{th}$ moment, and
assume:
\begin{itemize}
\item[(i)] $\lim_{N\to\infty}\E_{\empr_{x_0(N)}}(X_0^{2k-2})
=\E_{p_{X_0}}(X_0^{2k-2})<\infty$.
\item[(ii)] $\lim_{N\to\infty}\E_{\empr_{w(N)}}(W^{2k-2})=
\E_{\pW}(W^{2k-2})<\infty$.
\item[(iii)] $\lim_{N\to\infty} \E_{\empr_{q_0(N)}}( X^{2k-2})<\infty$.
\end{itemize}
Then, for any pseudo-Lipschitz function $\psi:\reals^2\to\reals$
of order $k$ and all $t\ge 0$, almost surely
\begin{eqnarray}
\lim_{N\to\infty}\f{1}{N}\sum_{i=1}^N\psi(h_i^{t+1},x_{0,i})
=\E\Big\{ \psi\big(\tau_t Z,X_0\big)\Big\}\, ,\\
\lim_{n\to\infty}\f{1}{n}\sum_{i=1}^n\psi(b_i^{t},w_{i})
=\E\Big\{ \psi\big(\sigma_t Z,W\big)\Big\}\, ,
\end{eqnarray}
where $X_0\sim p_{X_0}$ and $W\sim \pW$ are independent of $Z\sim \normal(0,1)$,
and $\sigma_t$, $\tau_t$ are determined by recursion
(\ref{eq:tau-sigma-recursion}).
\end{theorem}
%

%
%
\subsection{Corollary of Theorem \ref{thm:general}: AMP and Theorem \ref{thm:main}}

As already mentioned, the AMP algorithm (\ref{eq:dmm})
is a special case of recursion (\ref{eq:mpMain}).
The reduction is obtained by defining
\begin{eqnarray}
h^{t+1}& = &x_0-(A^*z^{t}+x^{t})\, ,\\
q^t& = & x^t-x_0\, ,\\
b^t & = & w-z^t\, ,\\
m^t & = & -z^t\, .
\end{eqnarray}
The functions $f_t$ and $g_t$ are given by
\begin{eqnarray}
f_t(s,x_0)=\eta_{t-1}(x_0-s)-x_0\, ,\;\;\;\;\;\;\;\;\;\;
g_t(s, w)=s-w\,,
\end{eqnarray}
and the initial condition is $q^0=-x_0$.
\begin{note}\label{nt:AMP}~
\begin{itemize}
\item[(a)] Although the recursions \eqref{eq:dmm} and \eqref{eq:mpMain}
are equivalent mathematically, only the former can be used as an algorithm.
Indeed the recursion \eqref{eq:mpMain} tracks the difference of
the current estimates $x^t$ from $x_0$, and is initialized using $x^0$
itself. The recursion \eqref{eq:mpMain} is only relevant
for mathematical analysis.
\item[(b)] Due to symmetry, for each $t$, all coordinates of the
vector $h^t$ have the same distribution
(similarly for $b^t$, $q^t$ and $m^t$).
\end{itemize}
\end{note}

%
%
\subsection{Proof of Theorem \ref{thm:main}}
First note that \eqref{eq:tau-sigma-recursion} reduces to
\[\tau_t^2=\sigma^2+\sigma_t^2=\sigma^2+\f{1}{\delta}\,\E\Big\{\Big(\eta_{t-1}(X_0+\tau_{t-1}Z)-X_0\Big)^2\Big\},\]
with $\tau_0^2 = \sigma^2+\delta^{-1}\E(X_0^2)$. The latter follows from
\[\sigma_0^2=\frac{1}{\delta}\lim_{N\to\infty} \frac{1}{N}\|q^0\|^2=\frac{1}{\delta}\E_{p_{X_0}}(X_0^2)\]
and $\tau_0^2=\sigma^2+\sigma_0^2$.  Also, by definition, $x^{t+1}=\eta_t(A^*b^t+x^t)=\eta_t(x_0-h^{t+1})$.  Therefore,
applying Theorem \ref{thm:general} to the function $
(h^t_i,x_{0,i})\mapsto\psi(\eta_{t-1}(x_{0,i}-h^t_i),x_{0,i})$ we obtain almost surely
\begin{align*}
\lim_{N\to\infty}\f{1}{N}\sum_{i=1}^N\psi(x_i^{t},x_{0,i})=\E\left\{\psi\big(\eta_{t-1}(X_0-\tau_{t-1}Z),X_0\big)\right\}\,,
\end{align*}
with $Z\sim \normal(0,1)$ independent of $X_0\sim p_{X_0}$, which yields the claim as $Z$ has the same distribution as $-Z$.
Note that since $\eta$ is Lipschitz continuous, when $\psi$ belongs to $\psL(k)$ then
$(h^t_i,x_{0,i})\mapsto\psi(\eta_{t-1}(x_{0,i}-h_i^t),x_{0,i})$ also belongs to $\psL(k)$.
%
%
\subsection{Definitions and notations}\label{subsec:definitions}

When the update equation for $h^{t+1}$ in \eqref{eq:mpMain} is used, all values of $b^0,\ldots,b^{t}$, $m^0,\ldots,m^{t}$, $h^1,\ldots,h^t$ and $q^0,\ldots,q^t$ have been previously calculated.
Hence, we can consider the distribution of $h^{t+1}$ conditioned on all these known variables and also conditioned on $x_0$ and $w$.  In particular, define $\sigal{S}_{t_1,t_2}$ to be the $\sigma$-algebra generated by $b^0,\ldots,b^{t_1-1}$, $m^0,\ldots,m^{t_1-1}$, $h^1,\ldots,h^{t_2}$, $q^0,\ldots,q^{t_2}$ and $x_0$ and $w$.
The basic idea of the proof is to compute the conditional distributions
$b^t|_{\sigal{S}_{t,t}}$ and $h^{t+1}|_{\sigal{S}_{t+1,t}}$.
This is done by characterizing the conditional distribution of the matrix $A$
given this filtration.

Regarding $h^{t}$ and $b^t$ as column vectors, the equations for $b^0,\ldots,b^{t-1}$ and $h^1,\ldots,h^{t}$ can be written in matrix form as:
\begin{align*}
\underbrace{\left[h^1+\xi_0q^0|h^2+\xi_1q^1|\cdots|h^t+\xi_{t-1}q^{t-1}\right]}_{X_t}&=A^*\underbrace{[m^0|\ldots|m^{t-1}]}_{M_t},\\
\underbrace{\left[b^0|b^1+\lambda_1m^0|\cdots|b^{t-1}+\lambda_{t-1}m^{t-2}\right]}_{Y_t}&=A\underbrace{[q^0|\ldots|q^{t-1}]}_{Q_t}.
\end{align*}
or in short $X_t=A^*M_t$ and $Y_t=AQ_t$ . Here and below we use
vertical lines to indicate columns of a matrix,
i.e. $[a_1|a_2|\dots|a_k]$ is the matrix with columns $a_1,\dots,a_k$.

We also introduce the notation $m_{\pl}^t$ for the projection of $m^t$ onto the column space of $M_t$ and define $m_{\perp}^t=m^t-m_{\pl}^t$. Similarly, define $q_\pl^t$ and $q_\perp^t$ to be the parallel and orthogonal projections of $q^t$ onto column space of $Q_t$.
In particular, let $\va=\va_t=(\alpha_0,\ldots,\alpha_{t-1})$ and $\vb=\vb_t=(\beta_0,\ldots,\beta_{t-1})$ be the vectors (in $\reals^t$) of coefficients for
these projections. i.e.,
\begin{align}
m_{\pl}^t = \sum_{i=1}^{t-1}\alpha_i m^i\,,~~~~~ q_{\pl}^t=\sum_{i=0}^{t-1}\beta_i q^i\,\label{eq:alpha-beta-defined}.
\end{align}
We will show in Section \ref{subsec:pf-lem-elephant} (cf. Corollary \ref{coro:beta-alpha-finite}) that for any fixed $t$ as $N$ goes to infinity the quantities $\beta_i$'s and $\alpha_j$'s have a finite limit.

Recall that $D^*$ denotes the transpose of the matrix $D$ and for a vector $u\in\reals^m$: $\< u\>=\sum_{i=1}^mu_i/m$.
Also, for vectors $u,v\in\reals^m$ we define the scalar product
\[
\< u,v\>\equiv\frac{1}{m}\sum_{i=1}^mu_iv_i\,.
\]

Given two random variables $X,Y$, and a $\sigma$-algebra
$\sigal{S}$, the notations
$X|_{\sigal{S}}\deq Y$ means that for
any integrable function $\phi$ and for any random variable $Z$ measurable
on $\sigal{S}$, $\E\{\phi(X)Z\}= \E\{\phi(Y)Z\}$.
In words we will say that $X$ is distributed as
(or is equal in distribution to) $Y$ \emph{conditional
on} $\sigal{S}$.
In case $\sigal{S}$ is the trivial $\sigma$ algebra we simply write
$X\deq Y$ (i.e. $X$ and $Y$ are equal in distribution).
For random variables $X,Y$ the notation $X\asequal Y$ means that $X$ and $Y$ are equal almost surely.

The large system limit will be denoted either as
$\lim_{N\to\infty}$ or as $\lim_{n\to\infty}$. It is understood that
either of the two dimensions can index the sequence of problems under
consideration, and that $n/N\to\delta$.  In the large system limit,
we use the notation $\order_t(1)$ to
represent a vector in $\reals^t$ (with $t$ fixed) such that all of its coordinates converge to 0 almost surely as $N\to\infty$.

Finally, we will use $\identity_{d\times d}$ to denote the $d\times d$
identity matrix (and drop the subscript when dimensions should be
clear from the context). Similarly, $\zeros_{n\times m}$ is used to denote the $n\times m$ zero matrix.
The indicator function of property ${\cal A}$
is denoted by $\indicator({\cal A})$ or $\indicator_{{\cal A}}$. The normal distribution
with mean $\mu$ and variance $v^2$ is $\normal(\mu,v^2)$.
%
%
\subsection{Main technical Lemma}

We prove the following more general result.
\begin{lemma}\label{lem:elephant}
Let $\{A(N)\}$, $\{q_0(N)\}_N$, $\{x_0(N)\}_N$ and $\{w(N)\}_N$ be sequences as in Theorem \ref{thm:general},  with $n/N\to\delta\in (0,\infty)$
and let $\{\sigma_t,\tau_t\}_{t\ge 0}$ be defined uniquely by the
recursion \eqref{eq:tau-sigma-recursion} with initialization
$\sigma_0^2=\delta^{-1}\lim_{n\to\infty}\<q^0,q^0\>$.  Then the following hold
for all $t\in\naturals\cup\{0\}$
\begin{itemize}
\item[$(a)$]
\begin{align}
h^{t+1}|_{\sigal{S}_{t+1,t}}&\deq \sum_{i=0}^{t-1}{\alpha_i}h^{i+1}+ {\tA}^*m_\perp^t+\tQ_{t+1}\order_{t+1}(1)\, ,\label{eq:main-lem-h-c}\\
b^t|_{\sigal{S}_{t,t}}&\deq \sum_{i=0}^{t-1}{\beta_i}b^{i} + {\tA}q_\perp^{t}+ \tM_t\order_t(1)\, ,\label{eq:main-lem-z-c}
\end{align}
where ${\tA}$ is an independent copy of $A$ and the matrix $\tQ_t$ ($\tM_t$) is such that its columns form an orthogonal basis for the column space of $Q_t$ ($M_t$) and $\tQ_t^*\tQ_t=N\,\identity_{t\times t}$ ($\tM_t^*\tM_t=n\,\identity_{t\times t}$).

\item[$(b)$] For all pseudo-Lipschitz functions $\phi_h,\phi_b:\reals^{t+2}\to\reals$ of order $k$
\begin{align}
\lim_{N\to\infty}\f{1}{N}\sum_{i=1}^N\phi_h(h_i^1,\ldots,h_i^{t+1},x_{0,i})
&\asequal\E\big\{\phi_h(\tau_0 Z_0,\ldots,\tau_tZ_{t},X_0)\big\}\, ,\label{eq:main-lem-h-a}\\
\lim_{n\to\infty}\f{1}{n}\sum_{i=1}^n\phi_b(b_i^0,\ldots,b_i^{t},w_i)
&\asequal\E\big\{\phi_b(\sigma_0\hat{Z}_0,\ldots,\sigma_t\hat{Z}_t,W)\big\}\,
,\label{eq:main-lem-z-a}
\end{align}
where $(Z_0,\ldots,Z_{t})$ and $(\hat{Z}_0,\ldots,\hat{Z}_{t})$
are two zero-mean gaussian vectors independent of $X_0$, $W$,
with $Z_i,\hat{Z}_i\sim \normal(0,1)$.

\item[$(c)$] For all $0\leq r,s\leq t$ the following equations hold and all limits exist, are bounded and have degenerate distribution
(i.e. they are constant random variables):
\begin{align}
\lim_{N\to\infty}\< h^{r+1},h^{s+1}\>&\asequal \lim_{n\to\infty}\< m^{r},m^{s}\>\, ,\label{eq:main-lem-h-b}\\
\lim_{n\to\infty}\< b^r,b^s\>&\asequal\f{1}{\delta}\lim_{N\to\infty}\< q^r,q^s\>\, .\label{eq:main-lem-z-b}
\end{align}

\item[$(d)$] For all $0\leq r,s\leq t$, and for any Lipschitz function
$\varphi:\reals^2\to\reals$ , the following equations hold and all limits exist, are bounded and have degenerate distribution (i.e. they are constant random variables):
\begin{align}
\lim_{N\to\infty}\< h^{r+1}, \varphi(h^{s+1},x_0)\>&\asequal \lim_{N\to\infty}\< h^{r+1},h^{s+1}\> \< \varphi'(h^{s+1},x_0)\>,\label{eq:main-lem-h-d}\\
\lim_{n\to\infty} \< b^{r},\varphi(b^s,w)\>&\asequal \lim_{n\to\infty}\< b^r,b^s\> \< \varphi'(b^s,w)\>\,. \label{eq:main-lem-z-d}
\end{align}
Here $\varphi'$ denotes derivative with respect to the first coordinate of $\varphi$.

\item[$(e)$] For $\ell = k-1$, the following hold
almost surely
\begin{align}
\lim\sup_{N\to\infty}\f{1}{N}\sum_{i=1}^N (h_i^{t+1})^{2\ell}&<\infty\,,\label{eq:main-lem-h-e}\\
\lim\sup_{n\to\infty}\f{1}{n}\sum_{i=1}^n (b_i^t)^{2\ell}&<\infty.\label{eq:main-lem-z-e}
\end{align}

\item[$(f)$] For all $0\leq r\le t$:
\begin{align}
\lim_{N\to\infty}\f{1}{N}\< h^{r+1},q^0\>&\asequal0\,.\label{eq:<ht,q0>=0}
\end{align}

\item[$(g)$] For all $0\leq r\leq t$ and $0\leq s \leq t-1$ the following limits exist, and there exist strictly positive constants $\lbq_r$ and $\lbm_s$ (independent of $N$, $n$) such that almost surely

\begin{align}
\lim_{N\to\infty}\< q^{r}_{\perp},q^{r}_{\perp}\>&>\lbq_r\,,\label{eq:main-lem-q-g}\\
\lim_{n\to\infty}\< m^{s}_\perp,m^{s}_\perp\>&>\lbm_s\,.\label{eq:main-lem-m-g}
\end{align}

\end{itemize}
\end{lemma}

\begin{note}
Equations (\ref{eq:main-lem-h-d}) and  (\ref{eq:main-lem-z-d}) have the form
of Stein's lemma \cite{Ste72} (cf. Lemma \ref{lem:stein} in Section \ref{subsec:properties}).
\end{note}

%
%

\subsubsection{Proof of Theorem \ref{thm:general}}
Assuming Lemma \ref{lem:elephant} is correct Theorem \ref{thm:general} easily follows. To be more precise, Theorem \ref{thm:general} is a obtained by applying Lemma \ref{lem:elephant}(b) to functions
$\phi_h(y_0,\ldots,y_t,x_{0,i})=\psi(y_t,x_{0,i})$ and $\phi_b(y_0,\ldots,y_t,w_{i})=\psi(y_t,w_{i})$.
\endproof

The rest of Section \ref{sec:Analysis} focuses on proof of Lemma \ref{lem:elephant}.
%
%
\subsection{Useful probability facts}\label{sec:ProbFacts}

Before embarking in the actual proof, it is convenient to summarize
a few facts that will be used repeatedly.

We will use the following strong law of large numbers (SLLN)
for triangular arrays of independent but not identically distributed
random variables.
The form stated below follows immediately from \cite[Theorem 2.1]{SLLN2}.
{
\begin{theorem}[SLLN, \cite{SLLN2}]\label{thm:SLLN}
Let $\{X_{n,i}: 1\le i\le n, n\ge 1\}$ be a triangular array of random
variables with $(X_{n,1},\dots,X_{n,n})$ mutually independent
with mean equal to zero for each $n$ and $n^{-1}\sum_{i=1}^n\E|X_{n,i}|^{2+\slln}\le cn^{\slln/2}$ for some $0<\slln<1$, $c<\infty$.
Then $\frac{1}{n}\sum_{i=1}^nX_{i,n}\to 0$ almost surely for $n\to\infty$.
\end{theorem}
\begin{note}
Theorem 2.1 in \cite{SLLN2} is stronger than what we state here. In Appendix \ref{ap:SLLN} we show how Theorem \ref{thm:SLLN} follows from it.
\end{note}
Next, we present an algebraic inequality that will be used in conjunction with Theorem \ref{thm:SLLN}. Its proof is provided in Appendix \ref{ap:alg-lemma}
\begin{lemma}\label{lem:alg-lemma}
Let $u_1,\ldots,u_n$ be a sequence of non-negative numbers. Then for all $\varepsilon>0$ the following holds
\[
\sum_{i=1}^n u_i^{1+\varepsilon} \le \left(\sum_{i=1}^n u_i\right)^{1+\varepsilon}\,.
\]
\end{lemma}

}
Next, we present a standard property of Gaussian matrices without proof.
\begin{lemma}\label{lem:prop-Gaussian-matrix} For any deterministic $u\in\reals^N$ and $v\in\reals^n$ with $\|u\|=\|v\|=1$ and a gaussian matrix ${\tA}$ distributed as $A$ we have
\begin{itemize}
\item[(a)] $v^* {\tA} u\deq Z/\sqrt{n}$ where $Z\sim \normal(0,1)$.

\item[(b)] $\lim_{n\to\infty} \|{\tA}u\|^2 = 1$ almost surely.

\item[(c)] Consider, for $d\le n$, a $d$-dimensional subspace $W$ of
$\reals^n$, an orthogonal basis $w_1,\ldots,w_d$ of $W$ with $\|w_i\|^2=n$ for $i=1,\ldots,d$, and the orthogonal projection $P_W$ onto $W$. Then for $D=[w_1|\ldots|w_d]$, we have $P_W A u\deq Dx$ with $x\in\reals^d$ that satisfies: $\lim_{n\to\infty}\|x\|\asequal 0$
(the limit being taken with $d$ fixed). Note that $x$ is $\order_d(1)$ as well.
\end{itemize}
\end{lemma}
%
\begin{lemma}[Stein's Lemma \cite{Ste72}]\label{lem:stein}
For jointly gaussian random variables $Z_1,Z_2$ with zero mean, and any function $\varphi:\reals\to\reals$ where $\E\{\varphi'(Z_1)\}$ and
$\E\{Z_1\varphi(Z_2)\}$ exist, the following holds
\[
\E\{Z_1\varphi(Z_2)\}=\cov(Z_1,Z_2)\, \E\{\varphi'(Z_2)\}\,.
\]
\end{lemma}
We will apply the following law of large numbers
to the sequence $\{x_0(N),w(N)\}_N$.
Its proof can be found in Appendix \ref{app:SLLN}.
\begin{lemma}\label{lem:SLLN4us}
Let $k\geq 2$ and consider a sequence of vectors $\{v(N)\}_{N\ge 0}$,
 whose
empirical distribution converges weakly to probability measure
$p_{V}$ on $\reals$  with bounded $k^{th}$ moment, and
assume $\E_{\empr_{v(N)}}(V^k)\to\E_{p_{V}}(V^k)$ as $N\to\infty$.
Then, for any pseudo-Lipschitz function $\psi:\reals\to\reals$ of order $k$:
\begin{eqnarray}
\lim_{N\to\infty}\f{1}{N}\sum_{i=1}^N\psi(v_{i}) \asequal\E\big[ \psi(V)\big]\,.\label{eq:empirical->p_X}
\end{eqnarray}
\end{lemma}
Next lemma is on reak convergence
of Lipschitz functions and its proof is in Appendix \ref{app:Almost proof}.
\begin{lemma}\label{lem:AlmostSmooth}
Let $F:\reals^2\to\reals$ be Lipschitz continuous and denote by
$F'(x,y)$ its derivative with respect to the first argument at
$(x,y)\in\reals^2$. Assume $(X_n,Y_n)$ is a sequence of random vectors
in $\reals^2$ converging in distribution to the random vector $(X,Y)$
as $n\to\infty$. Assume further that $X$ and $Y$ are independent
and that the distribution of $X$ is absolutely continuous with respect to
the Lebesgue measure. Then
\begin{eqnarray}
\lim_{n\to\infty}\E\{F'(X_n,Y_n)\}=\E\{F'(X,Y)\}\, .
\end{eqnarray}
\end{lemma}

It is useful to remember a standard formula for the conditional variance of
gaussian random variables.
\begin{lemma}\label{lem:gaussian-cond-variance}
Let $(Z_1,\ldots,Z_t)$ be a normal random vector with zero mean,
and assume that the covariance matrix of $(Z_1,\ldots,Z_{t-1})$
(denoted by $C$) is invertible. Then
\[
\var(Z_t|Z_1,\ldots,Z_{t-1}) =
\E\{Z_t^2\}-u^* C^{-1} u\,,
\]
where $u\in\reals^{t-1}$ is given by $u_i \equiv\E\{Z_tZ_i\}$.
\end{lemma}

An immediate consequence is the following fact, proven in
Appendix \ref{app:ConditionalVarianceProof}.
\begin{lemma}\label{lem:ConditionalVariance}
Let $Z_1,\ldots,Z_t$ be a sequence of jointly gaussian
random variables and let $c_1,\ldots,c_t$ be strictly positive constants
such that for all $i=1,\ldots,t$: $\var(Z_i|Z_1,\ldots,Z_{i-1})\ge c_i$.
Further assume $\E\{Z_i^2\}\le K$ for all $i$ and some constant $K$.
Let $Y$ be a random variable in the same probability space.

Finally let $\lip:\reals^2\to\reals$ be a Lipschitz function,
with $z\mapsto\lip(z,Y)$ non-constant with positive probability
(with respect to $Y$).

Then there exist a positive constant $c'_t$ (depending on $c_1,\dots, c_t$,
on $K$, on the random variable $Y$, and on the function $\lip$)
such that
\[
\E\{[\lip(Z_t,Y)]^2\}-u^* C^{-1} u > c_t'\,,
\]
where $u\in\reals^{t-1}$
is given by $u_i \equiv\E\left\{\lip(Z_t,Y)\lip(Z_i,Y)\right\}$,
and $C\in\reals^{t-1\times t-1}$ satisfies $C_{ij}\equiv\E\left\{\lip(Z_i,Y)\,\lip(Z_j,Y)\right\}$ for $1\leq i,j\leq t-1$.
\end{lemma}

\subsubsection{Linear algebra facts}

It is also convenient to recall some linear algebra
facts. The first one is proved in Appendix
\ref{app:InvertibleCovariance}.
\begin{lemma}\label{lem:InvertibleCovariance}
Let $v_1,\ldots,v_t$ be a sequence vectors in $\reals^{n}$ such that for all $i=1,\ldots, t$:
\[
\f{1}{n}\|v_i-P_{i-1}(v_i)\|^2\geq c
\]
for a positive constant $c$ and let $P_{i-1}$ be the orthogonal projector to the span of $v_1,\ldots,v_{i-1}$.
Then there is a constant $c'$
(depending only on $c$ and $t$),
such that the matrix $C\in\reals^{t\times t}$ with $C_{ij}=\<v_i,v_j\>$ satisfies
\[
\lambda_{\min}(C)\geq c'\,.
\]
\end{lemma}

The second one is just a direct consequence of the fact that
the mapping $S\mapsto\lambda_{\rm min}(S)$ is  continuous at any matrix $S$
that is invertible.
\begin{lemma}\label{lem:inverse-is-continuous}
Let $\{S_n\}_{n\geq 1}$ be a sequence of $t \times t$ matrices such that
$\lim\inf_{n\to\infty}\lambda_{\min}(S_n)>c$ for a positive constant $c$.  Also assume that $\lim_{n\to\infty}S_n=S_\infty$ where the
limit is element-wise.  Then, $\lambda_{\min}(S_\infty)\ge c$.
\end{lemma}

%
%
\subsection{Conditional distributions}\label{subsec:properties}

In order to calculate $b^t|_{\sigal{S}_{t,t}}$ and $h^{t+1}|_{\sigal{S}_{t+1,t}}$
we will characterize the conditional distributions
 $A|_{\sigal{S}_{t,t}}$ and $A|_{\sigal{S}_{t+1,t}}$.

\begin{lemma}\label{lem:Conditional A}
For $(t_1,t_2) = (t,t)$ or $(t_1,t_2)= (t+1,t)$, the
conditional distribution of the random matrix $A$
given the $\sigma$-algebra ${\sigal{S}_{t_1,t_2}}$, satisfies
\begin{align}
A|_{\sigal{S}_{t_1,t_2}}&\deq E_{t_1,t_2} + \cP_{t_1,t_2}(\tA) \label{eq:Conditional A}.
\end{align}
Here $\tA\deq A$ is a random matrix independent of ${\sigal{S}_{t_1,t_2}}$
and $E_{t_1,t_2}=\E(A|{\sigal{S}_{t_1,t_2}})$ is given by
\begin{align}
E_{t_1,t_2}&=Y_{t_1}(Q_{t_1}^*Q_{t_1})^{-1}Q_{t_1}^*+M_{t_2}(M_{t_2}^*M_{t_2})^{-1}X_{t_2}^*-M_{t_2}(M_{t_2}^*M_{t_2})^{-1}M_{t_2}^*Y_{t_1}(Q_{t_1}^*Q_{t_1})^{-1}Q_{t_1}^*\, .\label{eq:Et}
\end{align}
Further, $\cP_{t_1,t_2}$ is the orthogonal projector onto subspace $\mathrm{V}_{t_1,t_2}=\{A|AQ_{t_1}=0,A^*M_{t_2}=0\}$, defined by
\[
\cP_{t_1,t_2}(\tA) =P_{M_{t_2}}^\perp {\tA} P_{Q_{t_1}}^\perp.
\]
Here $P_{M_{t_2}}^\perp = \identity-P_{M_{t_2}}$, $P_{Q_{t_1}}^\perp=\identity-P_{Q_{t_1}}$, and $P_{Q_{t_1}}$, $P_{M_{t_2}}$ are orthogonal projector onto column spaces of $Q_{t_1}$ and $M_{t_2}$ respectively.
\end{lemma}
Recall the following well-known formula.
\begin{lemma}\label{lem:Conditional Gaussian}
Let $z\in\reals^{n}$ be a random vector with
i.i.d. $\normal(0,v^2)$ entries and let $\mrm{D}\in \reals^{m\times n}$ be a linear operator
with full row rank. Then for any constant vector $b\in \reals^{m}$ the distribution of $z$ conditioned on $\mrm{D}z=b$ satisfies:
\[
z|_{\mrm{D}z=b}\deq\mrm{D}^*(\mrm{D}\mrm{D}^*)^{-1}b+\mrm{P}_{\{\mrm{D}z=0\}}(\tilde{z})
\]
where $\mrm{P}_{\{\mrm{D}z=0\}}$ is the orthogonal projection  onto the subspace $\{\mrm{D}z=0\}$ and $\tilde{z}$ is a random vector of i.i.d. $\normal(0,v^2)$.
Moreover,  $\mrm{D}^*(\mrm{D}\mrm{D}^*)^{-1}b=\arg\min_z\left\{\|z\|^2|\mrm{D}z=b \right\}.$
\end{lemma}
\begin{proof}
The result is trivial if $\mrm{D}=[ \identity_{m\times m}|\zeros_{m\times(n-m)}]$.
For general $\mrm{D}$, it follows by invariance of the gaussian distribution under rotations.
Finally, using a least square calculation, it is simple to see that $\mrm{D}^*(\mrm{D}\mrm{D}^*)^{-1}b=\arg\min_z\{\|z\|^2|\mrm{D}z=b \}$.
\end{proof}
Lemma \ref{lem:Conditional A} follows from applying Lemma \ref{lem:Conditional Gaussian} to the operator $\mrm{D}$ that maps $A$ to $(AQ, M^*A)$.
A detailed proof of Lemma \ref{lem:Conditional A} appears in Section \ref{subsec:pf-conditional-A}.
Note that we can assume, without loss of generality $f$, $g$
to be non-constant as a function of their
first argument. If this is the case, it is easy to see that,
for finite values of $t$, the matrices $M_t^*M_t$ and $Q_t^*Q_t$ are
non-singular almost surely, and hence the above expressions are well defined.
\begin{lemma}\label{lem:Aqt,Amt}
The following holds
%
\begin{align}
E_{t+1,t}^*m^t&= X_{t}(M_t^*M_t)^{-1}M_{t}^*m_\pl^t + Q_{t+1}(Q_{t+1}^*Q_{t+1})^{-1}Y_{t+1}^*m_\perp^t,\label{eq:Et(mt)}\\
E_{t,t}q^t&= Y_t(Q_t^*Q_t)^{-1}Q_t^*q_\pl^t + M_t(M_t^*M_t)^{-1}X_t^*q_\perp^t.\label{eq:Et(qt)}
\end{align}
\end{lemma}
\begin{proof}
Write $m^t=m^t_\pl+m^t_\perp$. Using \eqref{eq:Et} and the fact that $M_t^*m^t_\perp=0$,  we obtain
\[
E_{t+1,t}^*m^t_\perp=Q_{t+1}(Q_{t+1}^*Q_{t+1})^{-1}Y_{t+1}^*m_\perp^t.
\]
On the other hand let $m^t_\pl=\sum_{i=0}^{t-1}\alpha_im^i=M_t\va$. Then using $A^*M_t=X_t$, Eq. \eqref{eq:Conditional A}, and $[\cP_{t+1,t}(\tA)]^*m_\pl^t=0$ we have,
\begin{align*}
E_{t+1,t}^*m^t_\pl&= Q_{t+1}(Q_{t+1}^*Q_{t+1})^{-1}Y_{t+1}^*m^t_\pl + X_t(M_{t}^*M_{t})^{-1}M_{t}^*m^t_\pl - Q_{t+1}(Q_{t+1}^*Q_{t+1})^{-1}Y_{t+1}^* m^t_\pl\\
&=X_t(M_{t}^*M_{t})^{-1}M_{t}^*m^t_\pl\,.
\end{align*}
Similarly, writing $q^t=q^t_\pl+q^t_\perp$, $q^t_\pl=Q_t\vb$, and using $X_t^*Q_t=M_t^*AQ_t=M_t^*Y_t $, $Q_t^*q^t_\perp=0$ we obtain \eqref{eq:Et(qt)}.
\end{proof}
%

%
%
\subsubsection{Proof of Lemma \ref{lem:Conditional A}}\label{subsec:pf-conditional-A}

Conditioning on $\sigal{S}_{t_1,t_2}$ is equivalent to conditioning on
the linear constraints $A Q_{t_1}= Y_{t_1}$ and $A^*M_{t_2}=X_{t_2}$.
To simplify the notation, just in Section \ref{subsec:pf-conditional-A},
we will drop all sub-indices $t_1,t_2$.
The expression (\ref{eq:Et}) for the conditional expectation
$E = \E\{A|_{\sigal{S}_{t_1,t_2}}\}$
 follows from Lemma \ref{lem:Conditional Gaussian} and the
following calculation for
\[
E=\arg\min_A\left\{\|A\|_F^2\bigg|AQ=Y,A^*M=X\right\}\, ,
\]
where $\|A\|_F$ denotes the Frobenius norm of matrix $A$.
We use Lagrange multipliers method to obtain this minimum. Consider the Lagrangian
\[
L(A,\Theta,\Gamma)=\|A\|_F^2+\big(\Theta,(Y-AQ)\big)+\big(\Gamma,(X-A^*M)\big),
\]
with $\Theta\in\reals^{n\times t_{1}}$, $\Gamma\in\reals^{N\times t_{2}}$
and $(A,B) \equiv \trace(AB^*)$ the usual scalar product among matrices.
Imposing the stationarity conditions yields
\begin{align}
2A&=\Theta Q^*+M\Gamma^*\label{eq:Lagrange}
\end{align}
Equation \eqref{eq:Lagrange} does not have a unique solution for the parameters $\Theta$ and $\Gamma$. In fact if $\Theta_0$, $\Gamma_0$ are a solution then for any $t_2\times t_1$ matrix $R$ the new parameters $\Theta_R=\Theta_0+MR$ and $\Gamma_R=\Gamma_0-QR^*$ satisfy $\Theta_RQ^*+M\Gamma_R^*=\Theta_0Q^*+M\Gamma_0^*=2A$.  In particular for $R_1=\Gamma_0^*Q(Q^*Q)^{-1}$ we have
$Q^*\Gamma_{R_1}=0$. Multiplying \eqref{eq:Lagrange} by $Q$ from right (using $\Theta_{R_1},\Gamma_{R_1}$) we have $2Y=\Theta_{R_1} {Q^*Q}$ or $\Theta_{R_1}=2Y(Q^*Q)^{-1}$.
Now multiplying \eqref{eq:Lagrange} by $M^*$ from left we obtain $2X^*=2M^*Y(Q^*Q)^{-1}Q^*+M^*M\Gamma_{R_1}^*$ which leads to $\Gamma_{R_1}^*=2(M^*M)^{-1}
\big[X^*-M^*Y(Q^*Q)^{-1}Q^*\big]$. From these we see that $E=\E(A|{\sigal{S}_{t_1,t_2}})$ satisfies:
\begin{align*}
E&=Y(Q^*Q)^{-1}Q^*+M(M^*M)^{-1}X-M(M^*M)^{-1}M^*Y(Q^*Q)^{-1}Q^*\,.
\end{align*}
%
Now we are left with the task of proving that
$\mc{P}_{t_1,t_2}({\tA})= P_{M}^\perp {\tA} P_{Q}^\perp$. We need to show that the linear operator $\mc{F}:A \mapsto P_{M}^\perp A P_{Q}^\perp$ satisfies
\begin{enumerate}
\item[(a)] $\mc{F}\circ\mc{F}=\mc{F}$.
\item[(b)] $\mc{F}(A)\in V = \{A|AQ_{t_1}=0,A^*M_{t_2}=0\}$.
\item[(c)] $\mc{F}(A) = A$ for $A\in V$
\item[(d)] $\mc{F}$ is symmetric. That is for all matrices $A,B$: $(\mc{F}(A),B)=( A,\mc{F}(B))$.
\end{enumerate}
Now we check $(a)$-$(d)$:

(a) is correct since
\[
\mc{F}\circ\mc{F}(A)=P_{M}^\perp P_{M}^\perp A P_{Q}^\perp P_{Q}^\perp=P_{M}^\perp A P_{Q}^\perp.
\]

(b) is correct since by definition of $\mc{F}(A)Q=P_{M}^\perp A P_{Q}^\perp Q=0$ and similarly $\mc{F}(A)^*M=0$.

(c) follows because
\[
\mc{F}(A) = A-P_{M}A-AP_{Q}+P_{M}AP_{Q},
\]
and each of the last three term vanishes either because $AQ=0$ or because $A^*M=0$.

(d) is correct because
\begin{align*}
(\mc{F}(A),B)
&=\trace\left(P_{M}^\perp A P_{Q}^\perp B^*\right)\\
&=\trace\left( A P_{Q}^\perp B^*P_{M}^\perp\right)=( A,\mc{F}(B) )\, .
\end{align*}
\endproof

%
%
\subsection{Proof of Lemma \ref{lem:elephant}}\label{subsec:pf-lem-elephant}

The proof is by induction on $t$. Let $\pty{H}_{t+1}$ be the property that  \eqref{eq:main-lem-h-c}, \eqref{eq:main-lem-h-a}, \eqref{eq:main-lem-h-b}, \eqref{eq:main-lem-h-d}, \eqref{eq:main-lem-h-e}, \eqref{eq:<ht,q0>=0}, and \eqref{eq:main-lem-q-g} hold. Similarly, let $\pty{B}_t$ be the property that \eqref{eq:main-lem-z-c}, \eqref{eq:main-lem-z-a}, \eqref{eq:main-lem-z-b}, \eqref{eq:main-lem-z-d},  \eqref{eq:main-lem-z-e} and \eqref{eq:main-lem-m-g} hold.  The inductive proof consists of the following four main steps.
\begin{enumerate}

\item $\pty{B}_0$ holds.

\item $\pty{H}_1$ holds.

\item If $\pty{B}_{r}$, $\pty{H}_{s}$ hold for all $r<t$ and $s\leq t$ then $\pty{B}_t$ holds.

\item If $\pty{B}_{r}$, $\pty{H}_{s}$ hold for all $r\leq t$ and $s\leq t$ then $\pty{H}_{t+1}$ holds.

\end{enumerate}

For each of these steps we will have to prove several properties
that we will denote by $(a)$, $(b)$, $(c)$, $(d)$, $(e)$ and $(g)$ according to their appearance in Lemma \ref{lem:elephant}.
For $\pty{H}$ we also need to prove a property $(f)$.

It is immediate to check that our claims become trivial
if $x\mapsto f_t(x,X_0)$ is constant (i.e. independent of $x$) almost surely
(with respect to $X_0\sim p_{X_0}$), or if $x\mapsto g_t(x,W)$ is constant
almost surely (with respect to $W\sim p_{W}$). We will therefore assume
that neither of these degenerate cases hold.
%
%
\subsubsection{Step 1: $\pty{B}_0$}

Note that $b^0=Aq^0$.
\begin{itemize}

\item[$(a)$] $\sigal{S}_{0,0}$ is generated by $x_0$, $q^0$ and $w$.
Also $q^0=q^0_\perp$ since $Q_0$ is an empty matrix. Hence
\[
b^0|_{\sigal{S}_{0,0}}= Aq^0_\perp.
\]

\item[$(b)$] Let $\phi_b:\reals^2\to\reals$ be a pseudo-Lipschitz function of order $k$.  Hence, $|\phi_b(x)|\leq L(1+\|x\|^k)$.
Given $q^0$, $w$, the random variable  $\sum_{i=1}^n\phi_b([Aq^0]_i,w_i)/n$ is a sum of independent random variables. By Lemma \ref{lem:prop-Gaussian-matrix}(a) $[Aq^0]_i\deq Z\|q^0\|/\sqrt{n}$ for $Z\sim\normal(0,1)$.
Hence, using
\[
\lim_{n\to \infty}\< q^0,q^0\> = \delta\sigma_0^2 <\infty,
\]
for all $p\geq 2$ there exist a constant $c_p$ such that $\E\{|(Aq^0)_i|^{p}\}=\< q^0,q^0\>^{p/2}\E|Z|^{p}<c_p$.
{ Therefore, in order to check conditions of Theorem \ref{thm:SLLN} for $X_{n,i}\equiv \phi_b(b_i^0,w_i)-\E_A\{\phi_b(b_i^0,w_i)\}$ for a $\slln$ in the interval $(0,1)$,
\begin{align}
\frac{1}{n}\sum_{i=1}^n\E|X_{n,i}|^{2+\slln}&=\frac{1}{n}\sum_{i=1}^n\E\big|\phi_b(b_i^0,w_i)-\E_A\{\phi_b(b_i^0,w_i)\}\big|^{2+\slln}\nonumber\\
&=\frac{1}{n}\sum_{i=1}^n\left|\E_{A,\tA}\left\{\phi_b([\tA q^0]_i,w_i)-\phi_b([A q^0]_i,w_i)\right\}\right|^{2+\slln}\nonumber\\
&\le\frac{1}{n}\sum_{i=1}^n\left|\E_{A,\tA}\left\{\phi_b([\tA q^0]_i,w_i)-\phi_b([A q^0]_i,w_i)\right\}\right|^{2+\slln}\nonumber\\
&\le\frac{L}{n}\sum_{i=1}^n\left|\E_{A,\tA}\left\{\big|[\tA q^0]_i-[Aq^0]_i\big|(1+|w_i|^{k-1}+|[\tA q^0]_i|^{k-1}+|[A q^0]_i|^{k-1})\right\}\right|^{2+\slln}\nonumber\\
&\le c'+\frac{L'c''}{n}\sum_{i=1}^n|w_i|^{(k-1)(2+\slln)}\label{eq:moment-cond}
\end{align}
where $\tA$ is an independent copy of $A$.  Now using Lemma \ref{lem:alg-lemma} for $u_i=|w_i|^{2(k-1)}$ and $\varepsilon = \slln/2$ we have
\[
\sum_{i=1}^n|w_i|^{(k-1)(2+\slln)}\leq \left(\sum_{i=1}^n|w_i|^{2(k-1)}\right)^{1+\slln/2}
\]
which combined with Eq. \eqref{eq:moment-cond} leads to
\begin{align*}
\frac{1}{n}\sum_{i=1}^n\E|X_{n,i}|^{2+\slln}&\le c'+L'c''n^{\slln/2}\left(\frac{1}{n}\sum_{i=1}^n|w_i|^{2(k-1)}\right)^{1+\slln/2}\\
&\le c''' n^{\slln/2}\,.
\end{align*}
The last inequality uses assumption on empirical moments of $w$.}
Therefore, we can apply Theorem \ref{thm:SLLN} to get
\[
\lim_{n\to\infty}\f{1}{n}\sum_{i=1}^n \left[\phi_b(b_i^0,w_i)-\E_A\{\phi_b(b_i^0,w_i)\}\right]\asequal 0.
\]
Hence, using Lemma \ref{lem:SLLN4us} for $v=w$ and for $\psi(w_i)=\E_Z\{
\phi_b(\|q^0\|Z/\sqrt{n},w_i)\}$ we get
\[
\lim_{n\to\infty}\f{1}{n}\sum_{i=1}^n\E_A[\phi_b(b_i^0,w_i)]\asequal\E\big\{\phi_b(\sigma_0 Z,W)\big\}.
\]
Note that $\psi$ belongs to $\psL(k)$ since $\phi_b$ belongs to $\psL(k)$.

\item[$(c)$] Using Lemma \ref{lem:prop-Gaussian-matrix}, conditioned on $q^0$,
\[
\lim_{n\to\infty}\< b^0,b^0\>=\lim_{n\to\infty}\f{\|Aq^0\|^2}{n}\asequal\lim_{N\to\infty}\f{\< q^0,q^0\>}{\delta}=\sigma_0^2.
\]

\item[$(d)$] Using $\pty{B}_0(b)$, and $\phi(x,w_i)=x\varphi(x,w_i)$ we obtain $\lim_{n\to\infty}\< b^0,\varphi(b^0,w)\>$ $\asequal$
$\E\{\sigma_0\hat{Z}\varphi(\sigma_0\hat{Z},W)\}$,
which is equal to $\sigma_0^2\E\{\varphi'(\sigma_0\hat{Z},W)\}$ using Lemma \ref{lem:stein}. Note that $x\varphi$ belongs $\psL(k)$.

By part $(b)$, the empirical distribution of $(b^0,w)$
(i.e. the probability distribution on $\reals^2$ that puts mass
$1/n$ on each point $(b_i^0,w_i)$, $i\in [n]$) converges weakly to the distribution
of $(\sigma_0\hat{Z},W)$.
Using Lemma \ref{lem:AlmostSmooth}, we get $\lim_{n\to\infty}\< \varphi'(b^0,w)\>\asequal\E\{\varphi'(\sigma_0\hat{Z},W)\}$.

\item[$(e)$] Similar to $(b)$, conditioning on $q^0$, the term $\sum_{i=1}^n([Aq^0]_i)^{2\ell}/n$ is sum of independent random variables
(namely, gaussians to the power $2\ell$) and $\E\{|[Aq^0]_i|^p\}
=\< q^0,q^0\>^{p/2}\E\{Z^{p}\}<c$ for a constant $c$. Therefore,
by Theorem \ref{thm:SLLN}, we get
\[
\lim_{n\to\infty}\f{1}{n}\sum_{i=1}^n \left[([Aq^0]_i)^{2\ell}-\E_A\{([Aq^0]_i)^{2\ell}\}\right]\asequal 0.
\]
But, $\f{1}{n}\sum_{i=1}^n \E_A\{([Aq^0]_i)^{2\ell}\}=\< q^0,q^0\>^{\ell}\E_Z(Z^{2\ell})<\infty$.

\item[$(g)$] This case is trivial since there is not $m^s$ with $s\leq t-1=-1$.


\end{itemize}

\endproof

%
%
\subsubsection{Step 2: $\pty{H}_1$}

Note that $h^1=A^*m^0-\xi_0 q^0$.
\begin{itemize}

\item[$(a)$] $\sigal{S}_{1,0}$ is generated by  $x_0$, $q^0$, $w$, $b^0$ and $m^0$. Also $m^0=m^0_\perp$ since $M_0$ is an empty matrix.
Applying Lemma \ref{lem:Conditional A} we have
\[A|_{\sigal{S}_{1,0}}\deq b^0\|q^0\|^{-2}(q^0)^*+\tA P_{q^0}^\perp.\]
Hence,
\[
h^1|_{\sigal{S}_{1,0}}\deq P_{q^0}^\perp
\tA^*m^0 + \delta\f{\< b_0,m_0\>}{\< q^0,q^0\>}q^0 -\xi_0 q^0.
\]
But using $\pty{B}_0(d)$ for $\varphi=g_0$
\[
\lim_{n\to\infty}\< b^0,m^0\>=\lim_{n\to\infty}\< b^0,g_0(b^0,w)\>\asequal \lim_{n\to\infty} \< b^0,b^0\> \< g_0'(b^0,w)\> \asequal \lim_{n\to\infty}\xi_0\f{\< q^0,q^0\>}{\delta}.\]
Therefore,
\[
h^1|_{\sigal{S}_{1,0}}\deq P_{q^0}^\perp \tA^*m^0 + \order_1(1) q^0.
\]
Also $\pty{B}_0(b)$, applied to the function $\phi_b(x,w)=g_0(x,w)^2$ gives
\begin{eqnarray}
\lim_{n\to\infty}\< m^0,m^0\>\asequal \E[g_0(\sigma_0Z,W)^2]=\tau_0^2<\infty\,.
\label{eq:Moved}
\end{eqnarray}
Thus,
\[
P_{q^0}^\perp\tA^*m^0 = \tA^*m^0 - P_{q^0}\tA^* m^0 = \tA^*m^0 +\order_1(1) \tq^0\, ,
\]
where the last estimate follows from Lemma \ref{lem:prop-Gaussian-matrix}(c) and \eqref{eq:Moved}. Finally,
\begin{align}
h^1|_{\sigal{S}_{1,0}}&\deq \tA^*m^0 + \order_1(1)q^0\,.\label{eq:H1Rep}
\end{align}

\item[$(c)$] Using \eqref{eq:H1Rep}, \eqref{eq:Moved}, and Lemma \ref{lem:prop-Gaussian-matrix}, we get
\[
\lim_{N\to\infty}\< h^1,h^1\>|_{\sigal{S}_{1,0}}\deq\lim_{N\to\infty}\f{\|\tA^* m_0\|^2}{N}\asequal\lim_{N\to\infty}\< m^0,m^0\>\asequal\tau_0^2\,.
\]

\item[$(e)$] First note that, conditioning on $\sigal{S}_{1,0}$,
\[\f{1}{N}\sum_{i=1}^N (h_i^1)^{2\ell} = \f{1}{N}\sum_{i=1}^N ([\tA^*m^0]_i + \order_1(1)q_i^0)^{2\ell} \leq \f{2^{2\ell}}{2} \f{1}{N}\sum_{i=1}^N\left\{ ([\tA^*m^0]_i)^{2\ell} + \order_1(1)(q_i^0)^{2\ell}\right\}.\]
By assumption, $\lim_{N\to\infty}\f{1}{N}\sum_{i=1}^N(q_i^0)^{2\ell}<\infty$ and finiteness of $\f{1}{N}\sum_{i=1}^N ([\tA^*m^0]_i)^{2\ell}$ can be established similar to $\pty{B}_0(e)$ for the sum of functions of
independent gaussians $\sum_{i=1}^N([\tA^*m^0]_i)^{2\ell}/N$.

\item[$(f)$] Using \eqref{eq:H1Rep} and Lemma \ref{lem:prop-Gaussian-matrix}(a) we have almost surely
\[
\lim_{N\to\infty}\< h^1,q^0\>\deq \lim_{N\to\infty} \frac{Z\|m^0\|\|q^0\|}{N\sqrt{n}}\asequal \lim_{N\to\infty} \frac{Z}{\sqrt{N}}\sqrt{\<m^0,m^0\>\<q^0,q^0\>} \asequal 0\,.
\]

\item[$(b)$] This proof uses again Eq.~(\ref{eq:H1Rep}) and
is very similar to the proof of $\pty{B}_0(b)$.  First we need to control
 the error term $\order_1(1)\tilde{q}^0=\order_1(1)q^0$. In other words we need to show
\[
\lim_{N\to\infty}\f{1}{N}\sum_{i=1}^N \left[\phi_h\left( [\tA^*m^0]_i + \order_1(1)q_i^0,x_{0,i}\right)-\phi_h\left( [\tA^*m^0]_i,x_{0,i}\right)\right]\asequal 0.
\]
To simplify the notation let $a_{i}=( [\tA^*m^0]_i + \order_1(1)q_i^0,x_{0,i})$ and $c_{i}=([\tA^*m^0]_i,x_{0,i})$.  Now, using the pseudo-Lipschitz property of $\phi_h$:
\[
|\phi_h(a_i)-\phi_h(c_i)|\leq L\{1+\max(\|a_i\|^{k-1},\|c_i\|^{k-1})\}
\, |q_i^0|\order_1(1).
\]
Using Cauchy-Schwartz inequality,
\[
\f{1}{N}\sum_{i=1}^N|\phi_h(a_i)-\phi_h(c_i)|\leq L\max\left(\f{1}{N}\sum_{i=1}^N\|a_i\|^{2k-2},\f{1}{N}\sum_{i=1}^N\|c_i\|^{2k-2}\right)^{1/2}\< q^0,q^0\>^{1/2}\, \order_1(1)\,.
\]
Hence, we only need to show $\f{1}{N}\sum_{i=1}^N\|a_i\|^{2k-2}<\infty$ and $\f{1}{N}\sum_{i=1}^N\|c_i\|^{2k-2}<\infty$ as $N\to\infty$.
But
\[
\f{1}{N}\sum_{i=1}^N\|a_i\|^{2k-2}=O\left(\f{1}{N}\sum_{i=1}^N|h_i^1|^{2k-2}+\f{1}{N}\sum_{i=1}^N|x_{0,i}|^{2k-2}\right)\,,
\]
which is bounded using part $(e)$ and the original assumption on $x_0$. Similarly, using $\f{1}{N}\sum_{i=1}^N\|q_i^0\|^{2k-2}<\infty$, we obtain $\f{1}{N}\sum_{i=1}^N\|c_i\|^{2k-2}<\infty$.

Thus, from here we consider $\th^1|_{\sigal{S}_{1,0}}\equiv \tA^*m^0$
whose components are distributed as $\|m^0\|Z/\sqrt{n}$ for
$Z$ a standard normal random variable,
and will follow the steps taken in $\pty{B}_0(b)$. Conditionally on $\sigal{S}_{1,0}$, we can apply Theorem \ref{thm:SLLN} to get
\[
\lim_{N\to\infty}\f{1}{N}\sum_{i=1}^N \left[\phi_h(\th_i^1,x_{0,i})-\E_{\tA}\{\phi_h(\th_i^1,x_{0,i})\}\right]\asequal 0\,.
\]
{ Note that a similar inequality to \eqref{eq:moment-cond} can be obtained here as well and then the condition of Theorem \ref{thm:SLLN} follows using Lemma \ref{lem:alg-lemma} and the assumed bound on empirical moments of $x_0$.}
Then, using Lemma \ref{lem:SLLN4us} for $v=x_0$ and $\psi(x_{0,i})=\E_{\tA}\{\phi_h(\th_i^1,x_{0,i})\}$,  we obtain
\[
\lim_{N\to\infty}\f{1}{N}\sum_{i=1}^n\E_{\tA}\left\{\phi_h(h_i^1,x_{0,i})\right\}=\lim_{N\to\infty}\E_Z\left\{\phi_h(\f{\|m^0\|}{\sqrt{n}}Z,X_0)\right\}\asequal\E\left\{\phi_h(\tau_0 Z,X_0)\right\}.
\]
The last equality used $\pty{B}_0(c)$.

\item[$(d)$] Using $\pty{H}_1(b)$ for $\phi_h(x,x_{0,i})=x\varphi(x,x_{0,i})$ we obtain $\lim_{N\to\infty}\< h^1,\varphi(h^1,x_{0})\>$ $\asequal$
$\E\{\tau_0 Z\varphi(\tau_0 Z,X_0)\}$,
which is equal to $\tau_0^2\,\E\{\varphi'(\tau_0{Z},X_0)\}$ using Lemma \ref{lem:stein}. On the other hand, in proof of $(b)$ we showed that $\lim_{N\to\infty}\< h^1,h^1\>\asequal\tau_0^2$.

By part $(b)$ the empirical distribution of
$(h^1,x_0)$ (i.e. the probability distribution on $\reals^2$ that puts mass
$1/N$ on each point $(h_i^1,x_{0,i})$, $i\in [N]$)
converges weakly to $(\tau_0Z,X_0)$.
By applying Lemma \ref{lem:AlmostSmooth} to the Lipschitz
function $\varphi$, we get $\lim_{N\to\infty}\< \varphi'(h^1,x_0)\>\asequal\E\{\varphi'(\tau_0{Z},X_0)\}$.

\item[$(g)$] Since $t=0$, and $q^0=q^0_\perp$ then the result follows from \eqref{eq:limit-<||q0||^2>} and that $\sigma_0^2>0$.

\end{itemize}
\endproof

%
%
\subsubsection{Step 3: $\pty{B}_{t}$}
This part is analogous to step 1 albeit more complex.

First we prove $(g)$.
\begin{itemize}

\item[$(g)$] Note that using induction hypothesis $\pty{B}_{t-1}(b)$  for
$\phi_b(b^r_i,b^s_i,w_i)=g_r(b^r_i,w_i)g_s(b^s_i,w_i)$,
$0\leq r,s\leq t-1$ we have almost surely
\begin{align}
\label{eq:<mr,ms-as-gaussian>}
\lim_{n\to\infty}\<m^{r},m^{s}\>=\E\left\{g_r(\sigma_r\hat{Z}_r,W)\,g_s(\sigma_s\hat{Z}_s,W)\right\}\,.
\end{align}
On the other hand
\begin{align}
\label{eq:norm-m_perp-1}
\<m^{t-1}_\perp,m^{t-1}_\perp\>=\<m^{t-1},m^{t-1}\>-\frac{(m^{t-1})^*M_{t-1}}{n}\left[\frac{M_{t-1}^*M_{t-1}}{n}\right]^{-1}\frac{M_{t-1}^*m^{t-1}}{n}\,.
\end{align}
But using induction hypothesis, we have $\lim_{n\to\infty}\<m^r_\perp,m^r_\perp\>>\lbm_r>0$ for all $r<t-1$. So using Lemma \ref{lem:InvertibleCovariance}, for large enough $n$ the smallest eigenvalue of matrix $M_{t-1}^*M_{t-1}/n$ is larger than a positive constant $c'$ that is independent of $n$. Hence, by Lemma \ref{lem:inverse-is-continuous} its inverse converges to an invertible limit. Thus, Eqs. \eqref{eq:<mr,ms-as-gaussian>} and \eqref{eq:norm-m_perp-1} lead to
\begin{align}
\label{eq:norm-m_perp-2}
\lim_{n\to\infty}\<m^{t-1}_\perp,m^{t-1}_\perp\>\asequal \E\left\{[g_{t-1}(\sigma_{t-1}\hat{Z}_{t-1},W)]^2\right\} - u^*C^{-1}u
\end{align}
with $u\in \reals^{(t-1)}$ and $C\in\reals^{(t-1)\times(t-1)}$ such that for $1\leq r,s\leq t-1$:
\[u_r=\E\left\{g_{r-1}(\sigma_{r-1}\hat{Z}_{r-1},W)\,g_{t-1}(\sigma_{t-1}\hat{Z}_{t-1},W)\right\}~,~~~
C_{rs}=\E\left\{g_{r-1}(\sigma_{r-1}\hat{Z}_{r-1},W)\,g_{s-1}(\sigma_{s-1}\hat{Z}_{s-1},W)\right\}\,.
 \]
 Now the result follows from Lemma \ref{lem:ConditionalVariance} provided that we show for gaussian random variables $\sigma_0\hat{Z}_0,\ldots,\sigma_{t-1}\hat{Z}_{t-1}$, all conditional variances $\var[\sigma_{r}\hat{Z}_{r}|\,\sigma_0\hat{Z}_0,\ldots,\sigma_{r-1}\hat{Z}_{r-1}]$ are strictly positive for $r=0,\ldots,t-1$.
To prove the latter first using the induction hypothesis $\pty{B}_{t-1}(b)$, we have for all $0\leq r\leq t-1$
\[
\lim_{n\to\infty} \<b^r_{\perp},b^r_\perp\>=\lim_{n\to\infty}\left( \<b^r,b^r\>-\frac{(b^{r})^*B_{r}}{n}\left[\frac{B_{r}^*B_{r}}{n}\right]^{-1}\frac{B_{r}^*b^{r}}{n}\ \right)\asequal \var[\sigma_{r}\hat{Z}_{r}\,|\,\sigma_0\hat{Z}_0,\ldots,\sigma_{r-1}\hat{Z}_{r-1}]\,.
\]
Similar as above we used the fact that for large enough $n$ the matrix $B_{r}^*B_{r}/n$ has a smallest eigenvalue greater than a positive constant to obtain the limit of its inverse. On the other hand using induction hypothesis $\pty{B}_r(c)$ we have almost surely
\begin{align}
\lim_{n\to\infty} \<b^r_{\perp},b^r_\perp\>&
=\lim_{n\to\infty}\left( \<b^r,b^r\>-\frac{(b^{r})^*B_{r}}{n}\left[\frac{B_{r}^*B_{r}}{n}\right]^{-1}\frac{B_{r}^*b^{r}}{n}\ \right)\nonumber\\
&=\frac{1}{\delta}\lim_{N\to\infty}\left( \<q^r,q^r\>-\frac{(q^{r})^*Q_{r}}{N}\left[\frac{Q_{r}^*Q_{r}}{N}\right]^{-1}\frac{Q_{r}^*q^{r}}{N}\ \right)\nonumber\\
&=\frac{1}{\delta}\lim_{N\to\infty} \<q^r_{\perp},q^r_\perp\>\,.\label{eq:b-perp-rep}
\end{align}
And now by induction hypothesis $\pty{H}_r(g)$ we have $\lim_{N\to\infty}\<q^r_{\perp},q^r_\perp\> > \lbq_r$. Hence the result follows.

\begin{corollary}\label{coro:beta-alpha-finite}  The vectors
\begin{align*}
\vec{\alpha}&=(\alpha_0,\dots,\alpha_{t-1})=\left[\frac{M_t^*M_t}{n}\right]^{-1}\frac{M_t^*\,m^t}{n}\,,\\
\vec{\beta}&=(\beta_0,\dots,\beta_{t-1})=\left[\frac{Q_t^*Q_t}{N}\right]^{-1}\frac{Q_t^*\,q^t}{N}
\end{align*}
have finite limits as $n$ and $N$ converge to $\infty$.
\end{corollary}
\begin{proof}
We can apply Lemma \ref{lem:InvertibleCovariance} to obtain that for large enough $n$ the smallest eigenvalue of $M_t^*M_t/n$ is larger than a positive constant $c'$. Hence by Lemma \ref{lem:inverse-is-continuous} its inverse has a finite limit. Similarly, we can apply induction hypothesis $\pty{H}_{t}(g)$ and Lemmas \ref{lem:InvertibleCovariance} and \ref{lem:inverse-is-continuous} to the matrix $Q_t^*Q_t/N$.
\end{proof}

\item[$(a)$] Recall definition of $Y_t$ and $X_t$ from Section \ref{subsec:definitions}.
\begin{align}
X_t=H_t+Q_t\Xi_t\,,~~~~Y_t=B_t+ [0| M_{t-1}]\Lambda_t\,,\label{eq:YX}
\end{align}
where $\Xi_t=\diag(\xi_0,\ldots,\xi_{t-1})$, $H_t=[h^1|\cdots|h^t]$, $B_t=[b^0|\cdots|b^{t-1}]$, and $\Lambda_t=\diag(\lambda_0,\ldots,\lambda_{t-1})$.
\begin{lemma}\label{lem:kt-ht Conditioned}
The following holds

(a) $h^{t+1}|_{\sigal{S}_{t+1,t}}\deq H_{t}(M_{t}^*M_{t})^{-1}M_{t}^*m^t_\pl+  P_{Q_{t+1}}^\perp {\tA}^*P_{M_{t}}^\perp m^t+Q_t\order_t(1)$.

(b) $b^t|_{\sigal{S}_{t,t}}\deq B_t (Q_t^*Q_t)^{-1}Q_t^*q^t_\pl+ P_{M_{t}}^\perp {\tA} P_{Q_{t}}^\perp q^t+M_t\order_t(1)$.

\end{lemma}
\begin{proof}
In light of Lemmas \ref{lem:Conditional A} and \ref{lem:Aqt,Amt} we have
\begin{align*}
h^{t+1}|_{\sigal{S}_{t+1,t}}&\deq X_{t}(M_{t}^*M_{t})^{-1}M_{t}^*m^t_\pl
+ Q_{t+1}(Q_{t+1}^*Q_{t+1})^{-1}Y_{t+1}^*m_\perp^t + P_{Q_{t+1}}^\perp {\tA}^*P_{M_{t}}^\perp m^t-\xi_tq^t,\\
b^t|_{\sigal{S}_{t,t}}&\deq Y_t(Q_t^*Q_t)^{-1}Q_t^*q^t_\pl + M_t(M_t^*M_t)^{-1}X_t^*q_\perp^t + P_{M_{t}}^\perp {\tA} P_{Q_{t}}^\perp q^t-\lambda_tm^{t-1}.
\end{align*}
Now using \eqref{eq:YX}, we only need to show
\begin{align*}
Q_{t}\Xi_t(M_{t}^*M_{t})^{-1}M_{t}^*m^t_\pl+Q_{t+1}(Q_{t+1}^*Q_{t+1})^{-1}Y_{t+1}^*m_\perp^t-\xi_tq^t&=Q_t\order_t(1),\\
[0| M_{t-1}]\Lambda_t(Q_t^*Q_t)^{-1}Q_t^*q^t_\pl+ M_t(M_t^*M_t)^{-1}X_t^*q_\perp^t -\lambda_tm^{t-1}&=M_t\order_t(1).
\end{align*}
Recall that $m^t_\pl=M_t\va$ and $q^t_\pl=Q_t\vb$. On the other hand $Y_{t+1}^*m_\perp^t=B_{t+1}^*m_\perp^t$ because $M_{t}^*
m_\perp^t=0$. Similarly, $X_t^*q_\perp^t=H_t^*q_\perp^t$. Hence we need to show
\begin{align}
Q_{t}\Xi_t\va+Q_{t+1}(Q_{t+1}^*Q_{t+1})^{-1}B_{t+1}^*m_\perp^t-\xi_tq^t&=Q_t\order_t(1)\label{eq:error(h)}\\
~[0|M_{t-1}] \Lambda_t\vb + M_t(M_t^*M_t)^{-1}H_t^*q_\perp^t-\lambda_tm^{t-1}&=M_t\order_t(1).\label{eq:error(z)}
\end{align}
Here is our strategy to prove \eqref{eq:error(z)} (proof of \eqref{eq:error(h)} is similar). The left hand side is a linear combination of vectors $m^0,\ldots,m^{t-1}$. For any $\ell=1,\ldots,t$ we will prove that the coefficient of
$m^{\ell-1}\in\reals^n$ converges to $0$.  This coefficient in the left hand side is equal to
\[\left[(M_t^*M_t)^{-1}H_t^*q_\perp^t\right]_\ell-\lambda_\ell(-\beta_\ell)^{\indicator_{\ell\neq t}}=
\sum_{r=1}^t \left[(\f{M_t^*M_t}{n})^{-1}\right]_{\ell,r} \f{\< h^r,q^t-\sum_{s=0}^{t-1}\beta_s q^s\>}{\delta}-\lambda_\ell(-\beta_\ell)^{\indicator_{\ell\neq t}}.
\]
To simplify the notation denote the matrix $M_t^*M_t/n$ by $G$.
Therefore,
\[
 \lim_{N\to\infty}\textrm{Coefficient of } m^{\ell-1}=
\lim_{N\to\infty}\left\{\sum_{r=1}^t (G^{-1})_{\ell,r} \< h^r,q^t-\sum_{s=0}^{t-1}\beta_s q^s\>\f{1}{\delta}-\lambda_\ell(-\beta_\ell)^{\indicator_{\ell\neq t}}\right\}.
\]
But using the induction hypothesis $\pty{H}_t(d)$ for $\varphi=f_1,\dots, f_t$, and $\pty{H}_t(f)$,
the term $\< h^r,q^t-\sum_{s=0}^{t-1}\beta_s q^s\>/\delta$ is almost surely equal to the
limit of
$\< h^r,h^t\>\lambda_t   - \sum_{s=0}^{t-1}\beta_s \< h^r,h^s\>\lambda_s$. This can be modified, using the induction hypothesis $\pty{H}_t(c)$, to
$\< m^{r-1},m^{t-1}\>\lambda_t   - \sum_{s=0}^{t-1}\beta_s \< m^{r-1},m^{s-1}\>\lambda_s$ almost surely, which can be written as $G_{r,t}\lambda_t   - \sum_{s=0}^{t-1}\beta_s G_{r,s}\lambda_s$.
Hence,
\begin{align*}
\lim_{N\to\infty}\textrm{Coefficient of } m^{\ell-1} &\asequal  \lim_{N\to\infty}\left\{\sum_{r=1}^t (G^{-1})_{\ell,r} [G_{r,t}\lambda_t   - \sum_{s=0}^{t-1}\beta_s G_{r,s}\lambda_s]-\lambda_\ell(-\beta_\ell)^{\indicator_{\ell\neq t}}\right\}\\
&\asequal  \lim_{N\to\infty}\left\{\lambda_t\indicator_{t=\ell}   - \sum_{s=0}^{t-1}\beta_s\lambda_s\indicator_{\ell=s} -\lambda_\ell(-\beta_\ell)^{\indicator_{\ell\neq t}}\right\}\\
&\asequal 0
\end{align*}
Notice that the above series of equalities hold because $G$ has, almost surely,
a non-singular limit as $N\to\infty$ as shown in point $(g)$ above.

Equation  \eqref{eq:error(h)} is proved analogously,
using $\xi_t = \< g'(b^t,w)\>$.
\end{proof}
The proof of Eq.~\eqref{eq:main-lem-z-c} follows immediately since the last lemma yields
\begin{align*}
b^t|_{\sigal{S}_{t,t}}&\deq \sum_{i=0}^{t-1}{\beta_i}b^{i} + {\tA}q_\perp^{t}- M_t(M_t^*M_t)^{-1}M_t^* \tA q_\perp^t+M_t\order_t(1)\,.
\end{align*}
Note that, using Lemma \ref{lem:prop-Gaussian-matrix}(c), as $n,N\to\infty$,
\begin{align*}
M_t(M_t^*M_t)^{-1}M_t^* \tA q_\perp^t&\deq\tM_t\order_t(1)\,,
\end{align*}
which finishes the proof since $\tM_t\order_t(1)+M_t\order_t(1)=\tM_t\order_t(1)$.
\item[$(c)$] For $r,s<t$ we can use the induction hypothesis. For $s=t, r<t$,
we can apply Lemma \ref{lem:kt-ht Conditioned}
to $b^t$ (proved above), thus obtaining
\begin{align}
\< b^t,b^r\>|_{\sigal{S}_{t,t}}&\deq
\sum_{i=0}^{t-1}{\beta_i}\< b^{i},b^r\> + \< P_{M_t}^\perp\tA q_\perp^{t},b^r\>+\sum_{i=0}^{t-1}o(1)\< m^i,b^r\>\,,\no
\end{align}
Note that, by induction hypothesis $\pty{B}_{t-1}(d)$ applied
to $\varphi=g_{t-1}$, and using the bound $\pty{B}_{t-1}(e)$ to control $\<b^i,b^r\>$, we deduce that each
term $\< m^i,b^r\>$ has a finite limit. Thus,
\[
\lim_{n\to\infty}\sum_{i=0}^{t-1}o(1)\< m^i,b^r\>\asequal 0.
\]
We can use Lemma \ref{lem:prop-Gaussian-matrix} for $\< P_{M_t}^\perp{\tA}q_\perp^{t},b^r\>=\< {\tA}q_\perp^{t},P_{M_t}^\perp b^r\>$
(recalling that $\tA$ is independent of $q_\perp^{t}, P_{M_t}^\perp b^r$) to obtain
\[
\< {\tA}q_\perp^{t},P_{M_t}^\perp b^r\>\deq \f{\|q_\perp^{t}\|\|P_{M_t}^\perp b^r\|}{N}\f{Z}{\sqrt{n}}\stackrel{{\almostsurely}}{\to}0
\]
where the last estimate uses the induction hypothesis $\pty{B}_{t-1}(c)$ and $\pty{H}_t(c)$ which imply, almost surely, for some constant $c$,
$\< P_{M_t}^\perp b^r,P_{M_t}^\perp b^r\>\le \< b^r,b^r\><c$ and $\< q_\perp^{t},q_\perp^{t}\>\le \< q^t,q^t\><c$ for all $N$ large enough.  Finally, using the induction hypothesis $\pty{B}_{t-1}(c)$ for each term of the form $\< b^{i},b^r\>$
(noting that $i,r\le t-1$) and Corollary \ref{coro:beta-alpha-finite} we have
\begin{align}
\lim_{n\to\infty}\< b^t,b^r\>&\asequal \lim_{n\to\infty}\f{1}{\delta}\sum_{i=0}^{t-1}{\beta_i}\< q^{i},q^r\>\no\\
&\asequal\lim_{n\to\infty}\f{1}{\delta}\< q_\pl^t,q^r\>\asequal\lim_{n\to\infty}\f{1}{\delta}\< q^t,q^r\>\no\, .
\end{align}
The last line uses the definition of $\beta_i$ and $q_\perp^t\perp q^r$.

For the case of $r=s=t$, similarly, we have
\begin{align}
\< b^t,b^t\>|_{\sigal{S}_{t,t}}&\deq \sum_{i,j=0}^{t-1}\beta_i\beta_j\< b^{i},b^j\> + \<  P_{M_t}^\perp \tA q_\perp^{t}, P_{M_t}^\perp \tA q_\perp^{t}\>+o(1).\no
\end{align}
The contribution of other terms is $o(1)$ because:
\begin{itemize}
\item $\<P_{M_t}^\perp \tA q_\perp^{t},M_t\order_t(1)\>=\<\tA q_\perp^{t},P_{M_t}^\perp M_t\order_t(1)\>=0$.
\item $\<\sum_{i=0}^{t-1}\beta_i b^i,M_t\order_t(1)\>=o(1)$, using Corollary \ref{coro:beta-alpha-finite} and induction hypothesis  $\pty{B}_{t-1}(d)$  for $\varphi = g_j$.

\item $\<\sum_{i=0}^{t-1}\beta_i b^i,P_{M_t}^\perp \tA q_\perp^{t}\>=o(1)$ follows from Lemma \ref{lem:prop-Gaussian-matrix} and Corollary \ref{coro:beta-alpha-finite}.
\end{itemize}
The arguments at the last two points are completely analogous to the one
carried out in the case $s=t$, $r<t$ above.

Now, using Lemma \ref{lem:prop-Gaussian-matrix},
\begin{align*}
\lim_{n\to\infty}\<  P_{M_t}^\perp \tA q_\perp^{t}, P_{M_t}^\perp \tA q_\perp^{t}\>&=\lim_{n\to\infty}\left[\<  \tA q_\perp^{t}, \tA q_\perp^{t}\>-\<  P_{M_t} \tA q_\perp^{t}, P_{M_t} \tA q_\perp^{t}\>\right]\\
&\asequal\lim_{n\to\infty}\left[\f{\< q_\perp^{t}, q_\perp^{t}\>}{\delta}-o(1)\right]\,.
\end{align*}
Hence, from the induction hypothesis  $\pty{B}_{t-1}(c)$,
\begin{align}
\lim_{n\to\infty}\< b^t,b^t\>|_{\sigal{S}_{t,t}}&\asequal \lim_{n\to\infty}\sum_{i,j=0}^{t-1}\beta_i\beta_j\f{\< q^{i},q^j\>}{\delta}
+ \lim_{n\to\infty}\f{\< q_\perp^{t}, q_\perp^{t}\>}{\delta}\no\\
&\asequal \lim_{n\to\infty}\f{\< q_\pl^{t},q_\pl^t\>}{\delta} + \lim_{n\to\infty}\f{\< q_\perp^{t},q_\perp^{t}\>}{\delta}\no\\
&\asequal \lim_{n\to\infty}\f{\< q^{t},q^t\>}{\delta}\no.
\end{align}

\item[$(e)$]
Conditioning on $\sigal{S}_{t,t}$ and using Lemma \ref{lem:kt-ht Conditioned}
(proved at point $(a)$ above), almost surely,
\begin{align*}
\sum_{i=1}^n\f{1}{n}(b_i^t)^{2\ell}& \le \f{C}{n}\sum_{i=1}^n (\sum_{r=0}^{t-1}\beta_r b_i^r)^{2\ell}+\f{C}{n}\sum_{i=1}^n([P_{M_t}^\perp\tA q_\perp^t]_i)^{2\ell}+o(1)\f{C}{n}\sum_{r=0}^{t-1}\sum_{i=1}^n([m^r]_i)^{2\ell}\, ,
\end{align*}
for some constant $C=C(\ell,t)<\infty$. We will bound each of the above
summands.

\begin{itemize}
\item The term $n^{-1}\sum_{i=1}^n (\sum_{r=0}^{t-1}\beta_r b_i^r)^{2\ell}$ is finite since we can write
\[
\f{1}{n}\sum_{i=1}^n (\sum_{r=0}^{t-1}\beta_r b_i^r)^{2\ell} = O\left(\sum_{r=0}^{t-1} \beta_r^{2\ell} \f{1}{n}\sum_{i=1}^n (b_i^r)^{2\ell}\right)\,,
\]
use Corollary \ref{coro:beta-alpha-finite} and induction hypothesis $\pty{B}_{t-1}(e)$ for each of $n^{-1}\sum_{i=1}^n (b_i^r)^{2\ell}$.

\item For the term $n^{-1}\sum_{i=1}^n([m^r]_i)^{2\ell}$ we use
\[
 (m_i^r)^2 = g_r(b_i^r,w_{i,0})^2= O\Big((b_i^r)^2+w_{i,0}^2+g(0,0)^2)\Big)\,,
\]
that follows from the Lipschitz assumption on $g_r$. Thus,
\begin{align}\label{eq:m2ell-rep}
\frac{1}{n}\sum_{i=1}^n(m^r_i)^{2\ell}=O\left( \f{1}{n}\sum_{i=1}^n(b_i^r)^{2\ell}+\f{1}{n}\sum_{i=1}^nw_{i,0}^{2\ell}+g(0,0)^{2\ell}\right)\,,
\end{align}
which has a finite limit almost surely, using the induction hypothesis
$\pty{B}_{t-1}(e)$ and the assumption on $w$.

\item The term $n^{-1}\sum_{i=1}^n([P_{M_t}^\perp\tA q_\perp^t]_i)^{2\ell}$ can be
written as
\[
\frac{1}{n}\sum_{i=1}^n([P_{M_t}^\perp\tA q_\perp^t]_i)^{2\ell}=O\Big(\frac{1}{n}\sum_{i=1}^n([\tA q_\perp^t]_i)^{2\ell}\Big)+O\Big(\frac{1}{n}\sum_{i=1}^n([P_{M_t}\tA q_\perp^t]_i)^{2\ell}\Big)\,.
\]
Now, $n^{-1}\sum_{i=1}^n([\tA q_\perp^t]_i)^{2\ell}$ has a finite limit using the same proof as in $\pty{B}_0(b)$ and the fact that $\lim_{n\to\infty}\< q^t_\perp,q^t_\perp\>\le \lim_{n\to\infty}\< q^t,q^t\><\infty$ almost surely.

Finally, for $n^{-1}\sum_{i=1}^n([P_{M_t}\tA q_\perp^t]_i)^{2\ell}$ using Lemma \ref{lem:prop-Gaussian-matrix} and Corollary \ref{coro:beta-alpha-finite}
we can write
\begin{align*}
P_{M_t}\tA q_\perp^t&\deq M_t\left[\frac{M_t^*M_t}{n}\right]^{-1}
\Big[\frac{Z_0\|m^0\|\|q_\perp^t\|}{n\sqrt{n}}
\Big|\cdots\Big| \frac{Z_{t-1}\|m^{t-1}\| \|q^t_\perp\|}{n\sqrt{n}}\Big]^*\\
&= \frac{1}{\sqrt{n}}\sum_{r=0}^{t-1}c_rm^rZ_r\,,
\end{align*}
where $Z_0,\ldots,Z_{t-1}$ are iid with distribution $\normal(0,1)$.
and $c_0,\dots ,c_r$ are allmost surely bounded for all $N$ large enough.
Therefore, almost surely,
\begin{align*}
\frac{1}{n}\sum_{i=1}^n([P_{M_t}\tA q_\perp^t]_i)^{2\ell}&\deq
\frac{1}{n}\sum_{i=1}^{n}
\Big(\frac{1}{\sqrt{n}}\sum_{r=0}^{t-1}c_rm_i^rZ_r\Big)^{2\ell}\\
& \le C\sum_{r=0}^{t-1}
\frac{1}{n}\sum_{i=1}^{n}\Big(\frac{1}{\sqrt{n}}m_i^rZ_r\Big)^{2\ell}\\
& \le C'\sum_{r=0}^{t-1}
\frac{1}{n}\sum_{i=1}^{n}(m_i^r)^{2\ell}\,.
\end{align*}
Now each term is finite using the
same argument as in Eq. \eqref{eq:m2ell-rep}.
\end{itemize}

\item[$(b)$]
Using part $(a)$ we can write
\[
\phi_b(b_i^0,\ldots,b_i^t,w_i)|_{\sigal{S}_{t,t}}\deq\phi_b\left(b_i^0,\ldots,b_i^{t-1},\left[\sum_{r=0}^{t-1}{\beta_r}b^{r} + {\tA}q_\perp^{t}+\tM_t\order_t(1)\right]_i,w_i\right).
\]
Similar to the proof of $\pty{H}_0(b)$ we can drop the error term $\tM_t\order_t(1)$. Indeed, defining
\begin{eqnarray*}
a_{i}&=&\left(b_i^0,\ldots,b_i^{t-1},\left[\sum_{r=0}^{t-1}{\beta_r}b^{r} + {\tA}q_\perp^{t}+\tM_t\order_t(1)\right]_i,w_i\right)\,,\\
c_{i}&=&\left(b_i^0,\ldots,b_i^{t-1},\left[\sum_{r=0}^{t-1}{\beta_r}b^{r} + {\tA}q_\perp^{t}\right]_i,w_i\right)\,,
\end{eqnarray*}
by the pseudo-Lipschitz assumption
\[
|\phi_b(a_{i})-\phi_b(c_{i})|\leq L\Big\{1+\max\big(\|a_{i}\|^{k-1},\|c_{i}\|^{k-1}\big)\Big\}\Big|\sum_{r=0}^{t-1}\tm_i^r\Big|\, \order_1(1).
\]
Therefore, using Cauchy-Schwartz inequality twice, we have
\begin{align}
\f{1}{n}\Big|\sum_{i=1}^n\phi_b(a_{i})-\sum_{i=1}^n\phi_b(c_{i})\Big|&\leq L\big[\max(\sum_{i=1}^n\f{\|a_{i}\|^{2k-2}}{n},\f{\sum_{i=1}^n\|c_{i}\|^{2k-2}}{n})\big]^{\f{1}{2}}\big[\sum_{r=0}^{t-1}t^{\f{1}{2}}\< \tm^r,\tm^r\>\big]^{\f{1}{2}} \order_1(1)\,.
\label{eq:CS}
\end{align}
Also note that
\[
\f{1}{n}\sum_{i=1}^n\|a_{i}\|^{2\ell}\leq (t+1)^{\ell}
\big\{\sum_{r=0}^{t}\f{1}{n}\sum_{i=1}^n (b_i^r)^{2\ell}+\f{1}{n}\sum_{i=1}^n (w_i)^{2\ell}\big\}\,,
\]
which is finite almost surely using the induction hypothesis  $\pty{B}_t(e)$
proved above and the assumption on $w$. The term $n^{-1}\sum_{i=1}^n\|c_{i}\|^{2\ell}$ is bounded almost surely since
\begin{align*}
\f{1}{n}\sum_{i=1}^n\|c_{i}\|^{2\ell}&\leq
\f{C}{n}\sum_{i=1}^n\|a_{i}\|^{2\ell}+C\sum_{r=0}^{t-1}
\f{1}{n}\sum_{i=1}^n(\tm^r)^{2\ell}\order_1(1)\\
& \leq
\f{C}{n}\sum_{i=1}^n\|a_{i}\|^{2\ell}+C'\sum_{r=0}^{t-1}
\f{1}{n}\sum_{i=1}^n(m^r)^{2\ell}\order_1(1)\, ,
\end{align*}
where the last inequality follows from the fact that $[M_t^*M_t/n]$
has almost surely a non-singular limit as $N\to\infty$, as proved in point $(g)$
above. Finally, for $r\le t-1$, each term
 $(1/n)\sum_{i=1}^n(\tm^r)^{2\ell}$ is bounded
using the induction hypothesis $\pty{B}_{t-1}(e)$, and the argument in
Eq.~\eqref{eq:m2ell-rep}.

Hence for any fixed $t$, \eqref{eq:CS} vanishes almost surely when $n$ goes to $\infty$.

Now given, $b^0,\ldots,b^{t-1}$,  consider the random variables
\[
\tilde{X}_{i,n}= \phi_b\left(b_i^0,\ldots,b_i^{t-1},\sum_{r=0}^{t-1}{\beta_r}b^{r}_i + ({\tA}q_\perp^{t})_i,w_i\right)
\]
and $X_{i,n}\equiv\tilde{X}_{i,n}-\E_{{\tA}}\{\tilde{X}_{i,n}\}$.
Proceeding as in Step 1, and using the pseudo-Lipschitz property of
$\phi$, it is easy to check the conditions of Theorem \ref{thm:SLLN}.
{ In particular, a similar inequality to \eqref{eq:moment-cond} can be obtained here as well and then the condition of Theorem \ref{thm:SLLN} follows using Lemma \ref{lem:alg-lemma} and induction hypothesis $\pty{B}_r(e)$ on the empirical moment bounds for $b^r$ for $r=0,\ldots,t-1$ and for $w$.}
We therefore get
\begin{multline}\label{eq:Ztb-halfway1}
\lim_{n\to\infty}\f{1}{n}\sum_{i=1}^n\bigg[\phi_b\Big(b_i^0,\ldots,b_i^{t-1},\big[\sum_{r=0}^{t-1}{\beta_r}b^{r} + {\tA}q_\perp^{t}\big]_i,w_i\Big)\\
-\E_{{\tA}}\Big\{\phi_b\Big(b_i^0,\ldots,b_i^{t-1},\big[\sum_{r=0}^{t-1}{\beta_r}b^{r} + {\tA}q_\perp^{t}\big]_i,w_i\Big)\Big\}\bigg]\asequal 0.
\end{multline}
Note that $[{\tA}q_\perp^{t}]_i$ is a gaussian random variable with variance $\|q_\perp^{t}\|^2/n$.
Further  $\|q_\perp^{t}\|^2/n$ converges to
a finite limit $\gamma_t^2$ almost surely as $N\to\infty$.
Indeed  $\|q_\perp^{t}\|^2/N =  \|q^{t}\|^2/N-\|q^t_{\pl}\|^2/N$.
By induction hypthesis  $\pty{H}_{t}(b)$ applied to the pseudo-Lipshitz
function $\phi_h(h^t_i,x_{0,i}) = f_t(h^t_i,x_{0,i})^2$,
$\|q^{t}\|^2/N = \<f_t(h^t,x_0),f_t(h^t,x_0)\>$ converges to a finite limit.
Further $\|q^t_{\pl}\|^2/N =\sum_{r,s=0}^{t-1}\beta_r\beta_r\<q^r,q^s\>$
also converges since the products $\<q^r,q^s\>$ do and the
coefficients $\beta_r$, $r\le t-1$
converge by Corollary  \ref{coro:beta-alpha-finite}.

Hence we can use induction hypothesis $\pty{B}_{t-1}(b)$ and
Corollary \ref{coro:beta-alpha-finite} for
\[
\widehat{\phi}_b(b_i^0,\ldots,b_i^{t-1},w_i)=\E_{Z}\Big\{\phi_b\Big(b_i^0,\ldots,b_i^{t-1},\sum_{r=0}^{t-1}{\beta_r}b_i^{r} + \f{\|q_\perp^{t}\|Z}{\sqrt{n}},w_i\Big)\Big\}\,,
\]
where $Z$ is an independent $\normal(0,1)$ random variable to show
\begin{multline}\label{eq:Ztb-halfway2}
\lim_{n\to\infty}\f{\sum_{i=1}^n\E_{{\tA}}\left\{\phi_b\left(b_i^0,\ldots,b_i^{t-1},\left[\sum_{r=0}^{t-1}{\beta_r}b^{r} + {\tA}q_\perp^{t}\right]_i,w_i\right)\right\}}{n}\\
\asequal\E\,\E_{Z}\Big\{\phi_b\Big(\sigma_0 Z_0,\ldots\sigma_{t-1}Z_{t-1},\sum_{r=0}^{t-1}{\beta_r}\sigma_{r} Z_{r} + \gamma_t\, Z,W\Big)\Big\}\, .
\end{multline}

Note that $\sum_{r=0}^{t-1}{\beta_r}\sigma_{r} Z_{r} + \gamma_t \,Z$ is gaussian.
All that we need, is to show that the variance of this gaussian is $\sigma_t^2$. But using a combination of \eqref{eq:Ztb-halfway1} and \eqref{eq:Ztb-halfway2}
for the pseudo-Lipschitz function $\phi_b(y_0,\ldots,y_t,w_i)=y_t^2$,
\begin{align}
\lim_{n\to\infty}\< b^t,b^t\>\asequal\E\Big\{\Big(\sum_{r=0}^{t-1}{\beta_r}\sigma_{r} Z_{r} + \gamma_t\, Z\Big)^2\Big\}.\label{eq:var-temp}
\end{align}
On the other hand in part $(c)$ we proved $\lim_{n\to\infty}\< b^t,b^t\>\asequal\lim_{n\to\infty}\delta^{-1}\< f(h^t,x_0),f(h^t,x_0)\>$. By induction hypothesis $\pty{H}_t(b)$ for the pseudo-Lipschitz function $\phi_h(y_0,\ldots,y_t,x_{0,i})=f(y_t,x_{0,i})^2$
we get $\lim_{n\to\infty}\delta^{-1}\< f(h^t,x_{0}),f(h^t,x_{0})\>\asequal\delta^{-1}\E\left\{f(\tau_{t-1}Z,X_0)^2\right\}$. So by definition
(\ref{eq:tau-sigma-recursion}),
both sides of \eqref{eq:var-temp} are equal to $\sigma_t^2$.

\item[$(d)$] In a manner very similar to the proof of $\pty{B}_0(d)$, using part $(b)$ for
the pseudo-Lipschitz function $\phi_b:\reals^{t+2}\to\reals$ that is given by
$\phi_b(y_0,\ldots,y_t,w_i)=y_t\varphi(y_s,w_i)$ we can obtain
\[
\lim_{n\to\infty}\< b^{t},\varphi(b^s,w)\>\asequal \E\left\{\sigma_t \hat{Z}_t \varphi(\sigma_s \hat{Z}_s,W)\right\}\,,
\]
for jointly gaussian $\hat{Z}_t, \hat{Z}_s$ with distribution $\normal(0,1)$. Using Lemma \ref{lem:stein}, this is almost surely equal to
$\cov(\sigma_t \hat{Z}_t,\sigma_s \hat{Z}_s)\E\{\varphi'(\sigma_s \hat{Z}_s,W)\}.$ By another application of part $(b)$ for $\phi_b(y_0,\ldots,y_t,w_i)=y_sy_t$ transforms $\cov(\sigma_t \hat{Z}_t,\sigma_s \hat{Z}_s)$ to $\lim_{n\to\infty}\< b^t,b^s\>$. Similar to $\pty{B}_0(d)$ we can use Lemma \ref{lem:AlmostSmooth} to transform $\E\{\varphi'(\sigma_s \hat{Z}_s,W)\}$  to $\lim_{n\to\infty}\< \varphi'(b^t,w)\>$ almost surely. This finishes the proof of $(d)$.

\end{itemize}

%
%
\subsubsection{Step 4: $\pty{H}_{t+1}$}
Due to symmetry, proof of this step is very similar to the proof of step 3 and we present only some differences.
\begin{itemize}

\item[$(g)$] This part is very similar to the one of $\pty{B}_t(g)$.

\item[$(a)$] To prove Eq.~\eqref{eq:main-lem-h-c} we use Lemma \ref{lem:kt-ht Conditioned}(a) as for $\pty{B}_{t}(a)$ to obtain
\begin{align*}
h^{t+1}|_{\sigal{S}_{t+1,t}}&\deq \sum_{i=0}^{t-1}{\alpha_i}h^{i+1}+ {\tA}^*m_\perp^t- Q_{t+1}(Q_{t+1}^*Q_{t+1})^{-1}Q_{t+1}^*\tA^*m_{\perp}^t+Q_t\order_t(1)\,.
\end{align*}
Now, using Lemma \ref{lem:prop-Gaussian-matrix}(c), as $n,N\to\infty$,
\begin{align*}
Q_{t+1}(Q_{t+1}^*Q_{t+1})^{-1}Q_{t+1}^*\tA^*m_{\perp}^t&\deq\tQ_{t+1}\order_t(1)\,
\end{align*}
which finishes the proof since $\tQ_{t+1}\order_t(1)+Q_t\order_t(1)=\tQ_{t+1}\order_t(1)$.

\item[$(c)$] For $r,s<t$ we can use induction hypothesis. For $s=t, r<t$, very similar to the proof of $\pty{B}_t(a)$,
\begin{align}
\< h^{t+1},b^{r+1}\>|_{\sigal{S}_{t+1,t}}&\deq \sum_{i=0}^{t-1}{\alpha_i}\< h^{i+1},h^{r+1}\> +
\< P_{Q_{t+1}}^{\perp} \tA^* m_\perp^{t},h^{r+1}\>+\sum_{i=0}^{t-1}o(1)\< q^i,h^{r+1}\>.\no
\end{align}
Now, by induction hypothesis $\pty{H}_{t}(d)$, for $\varphi=f$, each term $\< q^i,h^{r+1}\>$ has a finite limit. Thus,
\[
\lim_{N\to\infty}\sum_{i=0}^{t-1}o(1)\< q^i,h^{r+1}\>\asequal0.
\]
We can use induction hypothesis $\pty{H}_{r+1}(c)$ or $\pty{H}_{i}(c)$ for each term of the form $\< h^{i},h^{r+1}\>$ and use Lemma \ref{lem:prop-Gaussian-matrix} for $\< {\tA}^*m_\perp^{t},P_{Q_{t+1}}^{\perp} h^{r+1}\>$ to obtain
\begin{align}
\lim_{N\to\infty}\< h^{t+1},h^{r+1}\>&\asequal \lim_{N\to\infty}\sum_{i=0}^{t-1}{\alpha_i}\< m^{i},m^r\>\no\\
&\asequal\lim_{N\to\infty}\< m_\pl^t,m^r\>\asequal\lim_{N\to\infty}\< m^t,m^r\>\no\, ,
\end{align}
Where the last line uses the definition of $\alpha_i$ and $m_\perp^t\perp m^r$.

For the case of $r=s=t$, we have
\begin{align}
\< h^{t+1},h^{t+1}\>|_{\sigal{S}_{t+1,t}}&\deq \sum_{i,j=0}^{t-1}\alpha_i\alpha_j\< h^{i+1},h^{j+1}\> + \< P_{Q_{t+1}}^{\perp} \tA^* m_\perp^{t},P_{Q_{t+1}}^{\perp} \tA^* m_\perp^{t}\>+o(1).\no
\end{align}
Note that we used similar argument as in proof of $\pty{B}(c)$ to show the contribution of all products of the form $\< Q_t\order_t(1),\cdot\>$ and  $\< P_{Q_{t+1}}^{\perp} \tA^* m_\perp^{t},h^{i+1}\>$ a.s. tend to $0$.
Now, using induction hypothesis and Lemma \ref{lem:prop-Gaussian-matrix}

\begin{align}
\lim_{N\to\infty}\< h^{t+1},h^{t+1}\>|_{\sigal{S}_{t+1,t}}&\asequal \lim_{n\to\infty}\sum_{i,j=0}^{t-1}\alpha_i\alpha_j\< m^{i},m^j\> + \lim_{N\to\infty}\f{1}{N\delta}\|m_\perp^{t}\|^2\no\\
&\asequal \lim_{n\to\infty}\< m_\pl^{t},m_\pl^t\> + \lim_{n\to\infty}\< m_\perp^{t},m_\perp^{t}\>\no\\
&\asequal \lim_{n\to\infty}\< m^{t},m^t\>\no.
\end{align}

\item[$(e)$] This part is very similar to $\pty{B}_t(e)$.

\item[$(f)$]
Using $\pty{H}_t(a)$ and Lemma \ref{lem:prop-Gaussian-matrix}(a) we have almost surely
\[
\lim_{N\to\infty}\< h^{t+1},q^0\>\deq \lim_{N\to\infty}\frac{Z\|m^t_\perp\|\|q^0\|}{\sqrt{n}N}+\sum_{i=0}^{t-1}\lim_{N\to\infty}\alpha_i\< h^{i+1},q^0\>.
\]
%
But this limit is $0$ almost surely, using the induction hypothesis $\pty{H}_r(e)$ for $r<t$ and $\pty{B}_t(c)$.

\item[$(b)$] Using part $(a)$ we can write
\[
\phi_h(h_i^1,\ldots,h_i^{t+1},x_{0,i})|_{\sigal{S}_{t+1,t}}\deq\phi_h\left(h_i^1,\ldots,h_i^{t},\left[\sum_{r=0}^{t-1}{\alpha_r}h^{r+1} + {\tA}^*m_\perp^{t}+\tQ_{t+1}\order_{t+1}(1)\right]_i,x_{0,i}\right).
\]
Similar to proof of $\pty{B}_t(b)$ we can drop the error term $\tQ_{t+1}\order_{t+1}(1)$. Now given, $h^1,\ldots,h^{t}$,  consider the random variables
\[
\tilde{X}_{i,N}= \phi_h\left(h_i^1,\ldots,h_i^{t},\sum_{r=0}^{t-1}{\alpha_r}h^{r+1}_i + ({\tA}^*m_\perp^{t})_i,x_{0,i}\right)
\]
and $X_{i,N}\equiv\tilde{X}_{i,N}-\E_{{\tA}}\{\tilde{X}_{i,N}\}$.
Proceeding as in Step 2, and using the pseudo-Lipschitz property of
$\phi_h$, it is easy to check the conditions of Theorem \ref{thm:SLLN}.
We therefore  get
\begin{multline}\label{eq:Ht+1b-halfway1}
\lim_{N\to\infty}\f{1}{N}\sum_{i=1}^N\Bigg(\phi_h\Big(h_i^1,\ldots,h_i^{t},\big[\sum_{r=0}^{t-1}{\alpha_r}h^{r+1} + {\tA}^*m_\perp^{t}\big]_i,x_{0,i}\Big)\\
-\E_{{\tA}}\bigg\{\phi_h\Big(h_i^1,\ldots,h_i^{t},\big[\sum_{r=0}^{t-1}{\alpha_r}b^{r+1} + {\tA}^*m_\perp^{t}\big]_i,x_{0,i}\Big)\bigg\}\Bigg)\asequal 0.
\end{multline}
Note that $[{\tA}^*m_\perp^{t}]_i$ is a gaussian random variable with variance $\|m_\perp^{t}\|^2/n$. Hence we can use induction hypothesis $\pty{H}_{t}(b)$ for
\[
\widehat{\phi}_h(h_i^1,\ldots,h_i^{t},x_{0,i})=\E_{Z}\bigg\{\phi_h\Big(h_i^1,\ldots,h_i^{t},\sum_{r=0}^{t-1}{\alpha_r}h_i^{r+1} + \f{\|m_\perp^{t}\|Z}{\sqrt{n}},x_{0,i}\Big)\bigg\}\,,
\]
where $Z$ is an independent $\normal(0,1)$ random variable, to show
\begin{multline}\label{eq:Ht+1b-halfway2}
\lim_{N\to\infty}\f{\sum_{i=1}^N\E_{{\tA}}\left\{\phi_h\left(h_i^1,\ldots,h_i^{t},\left[\sum_{r=0}^{t-1}{\alpha_r}b^{r+1} + {\tA}^*m_\perp^{t}\right]_i,x_{i,0}\right)\right\}}{N}\\
\asequal\E\,\E_{Z}\left\{\phi_h\left(\tau_0 Z_0,\ldots\tau_{t-1}Z_{t-1},\sum_{r=0}^{t-1}{\alpha_r}\tau_{r} Z_{r} + \f{\|m_\perp^{t}\|Z}{\sqrt{n}},X_0\right)\right\}\,.
\end{multline}

Note that $\sum_{r=0}^{t-1}{\alpha_r}\tau_{r} Z_{r} + n^{-1/2}\|m_\perp^{t}\|Z$ is gaussian.
All that we need, is to show that the variance of this gaussian is $\tau_t^2$. But using combination of \eqref{eq:Ht+1b-halfway1} and \eqref{eq:Ht+1b-halfway2}
for the pseudo-Lipschitz function $\phi_h(y_0,\ldots,y_t,x_{0,i})=y_t^2$,
\begin{align}
\lim_{N\to\infty}\< h^{t+1},h^{t+1}\>\asequal\E\left\{\left(\sum_{r=0}^{t-1}{\alpha_r}\tau_{r} Z_{r} + \f{\|m_\perp^{t}\|Z}{\sqrt{n}}\right)^2\right\}.\label{eq:var-temp2}
\end{align}
On the other hand in part (c) we proved $\lim_{N\to\infty}\< h^{t+1},h^{t+1}\>\asequal\lim_{N\to\infty}\< g_t(b^t,w),g_t(b^t,w)\>$.

By the induction hypothesis $\pty{B}_t(b)$ for the pseudo-Lipschitz function $\phi_b(y_0,\ldots,y_t,w)=g_t(y_t,w)^2$
we get $\lim_{n\to\infty}\< g_t(b^t,w),g_t(b^t,w)\>\asequal\E\{g_t(\sigma_{t}Z,W)^2\}$. So by the definition
(\ref{eq:SERecursion}),
both sides of \eqref{eq:var-temp2} are equal to $\tau_t^2$.

\item[$(d)$] This is very similar to the proof of $\pty{B}_t(d)$. For the pseudo-Lipschitz function $\phi_h:\reals^{t+2}\to\reals$ that is given by
$\phi_h(y_1,\ldots,y_{t+1},x_{0,i})=y_{t+1}\varphi(y_{s+1},x_{0,i})$ we can use part $(a)$ to obtain
\[
\lim_{N\to\infty}\< h^{t+1},\varphi(b^{s+1},x_0)\>\asequal \E\{\tau_t Z_t \varphi(\tau_s Z_s,X_0)\}\,,
\]
for jointly gaussian $Z_t, Z_s$ with distribution $\normal(0,1)$. Using Lemma \ref{lem:stein}, this is almost surely equal to
$\cov(\tau_t Z_t,\tau_s Z_s)\E\{\varphi'(\tau_s Z_s,X_0)\}.$ And another application of part $(b)$ for $\phi_h(y_1,\ldots,y_{t+1},x_{i,0})=y_{s+1}y_{t+1}$ transforms $\cov(\tau_t {Z}_t,\tau_s {Z}_s)$ to $\lim_{N\to\infty}\< h^{t+1},h^{s+1}\>$. Similar to $\pty{H}_1(d)$ using Lemma \ref{lem:AlmostSmooth}, $\E\{\varphi'(\tau_s {Z}_s,X_0)\}$ can be transformed to $\lim_{N\to\infty}\< \varphi'(h^{t+1},x_0)\>$ almost surely. This finishes the proof of (d).
\end{itemize}

%
%
\subsection{Proof of Corollary \ref{coro:Decoupling}}
\label{sec:CoroProof}

First notice that the statement to be proved is equivalent to the following
claim. The joint distribution of
$(x_{J(1)}^{t},\dots,x_{J(\ell)}^t,
x_{0,J(1)},\dots,x_{0,J(\ell)})$, for
$J(1),\dots,J(\ell)\in [N]$ uniformly random subset of distinct indices,
converges weakly to to the distribution of
$(\widehat{X}_1,\dots \widehat{X}_\ell,X_{0,1},\dots,X_{0,\ell})$.
By general theory of weak convergence, it is therefore sufficient to
check Eq.~(\ref{eq:Kconv}) for functions of the form
\begin{eqnarray}
\psi(x_{1},\dots,x_{\ell},
y_{1},\dots,y_{\ell}) = \psi_1(x_{1},y_1)\cdots\psi_\ell(x_{\ell},y_{\ell})\, ,
\end{eqnarray}
for $\psi_i:\reals^2\to\reals$ Lipschitz and bounded.
This case follows immediately from Theorem \ref{thm:main}
once we notice that
\begin{eqnarray}
{\sf E}\, \psi(x_{J(1)}^{t},\dots,x_{J(\ell)}^t,
x_{0,J(1)},\dots,x_{0,J(\ell)}) = \prod_{s=1}^\ell
\Big(\frac{1}{N}\sum_{i=1}^N\psi_s(x^t_i,x_{0,i})\Big)+O(1/N)\, .
\end{eqnarray}
\endproof
%
%

\section{Symmetric Case}
\label{sec:Symmetric}

Let $k\ge2$, $G=A^*+A$ with $A\in\reals^{N\times N}$, and
assume that the entries of $A$ are i.i.d. $\normal(0,(2N)^{-1})$.  Also let $f: \reals\to\reals$ be a Lipschitz continuous function.  Start with $m^0$ and $m^1$ in $\reals^N$ where $m^0=\zeros_{N\times 1}$ and $m^1$ is a fixed deterministic vector in $\reals^N$ with $\lim\sup_{N\to\infty}N^{-1}\sum_{i=1}^N (m_{1,i})^{2k-2}< \infty$, and proceed by the following iteration
\begin{align}
h^{t+1}&=Gm^{t}-\lambda_t m^{t-1},\label{eq:sym}\\
m^t&=f(h^t)\no
\end{align}
where $\lambda_t=\< f'(h^t)\>$. Now let $\tau_1^2=\lim_{N\to\infty}\< m_1,m_1\>$, and define recursively for $t\ge 1$,
\begin{align}
\tau_{t+1}^2&  =\E\left\{[f(\tau_tZ)]^2\right\}\, ,\label{eq:SERecursionSym}
\end{align}
with $Z\sim \normal(0,1)$.

\begin{theorem}\label{thm:mainSym}
Let $\{A(N)\}_N$ be a sequence of matrices $A\in\reals^{N\times N}$
indexed by $N$, with i.i.d. entries $A_{ij}\sim \normal(0,1/(2N)^{-1})$.
Then, for any pseudo-Lipschitz function $\psi:\reals\to\reals$ of order $k$ and all $t\in\naturals$, almost surely
\begin{eqnarray}
\lim_{N\to\infty}\f{1}{N}\sum_{i=1}^N\psi(h_i^{t+1})
=\E\left[ \psi(f(\tau_{t} Z))\right]\, .
\end{eqnarray}
\end{theorem}
\begin{note}
This theorem was proved by Bolthausen in the case $f(x) = \tanh(\beta x+h)$
and $\<m^1,m^1\>=\tau_*^2$, for $\tau_*^2$ the
fixed point of the recursion \eqref{eq:SERecursionSym}.
The general proof is very similar to the one of Theorem \ref{thm:general},
and exploits the same conditioning trick. We omit it to avoid repetitions.
\end{note}

When we are calculating $h^{t+1}$, all values $h^1,\ldots,h^t$ and hence $m^1,\ldots,m^t$ are known to us. Denote the $\sigma$-algebra generated by all of these random variables by $\sigal{U}_t$. Moreover, use the following compact formulation for \eqref{eq:sym}.
\begin{align*}
\underbrace{\left[h^2|h^3+\lambda^2m^1|\cdots|h^t+\lambda^{t-1}m^{t-2}\right]}_{Y_{t-1}}&=G\underbrace{[m^1|\ldots|m^{t-1}]}_{M_{t-1}},
\end{align*}
The analogue of Lemma \ref{lem:elephant} is the following.
\begin{lemma}\label{lem:elephantSym}
Let $\{A(N)\}_N$ be a sequence of random matrices as in
Theorem \ref{thm:mainSym}. Then the following hold
for all $t\in\naturals$
\begin{itemize}
\item[$(a)$]
\begin{align}
h^{t+1}|_{\sigal{U}_{t}}&\deq \sum_{i=1}^{t-1}{\alpha_i}h^{i+1}+ {\tG}m_\perp^t+\tM_{t-1}\order_t(1)\,,
\end{align}
where ${\tG}$ is an independent copy of $G$ and coefficients $\alpha_i$ satisfy
$m_\pl^t=\sum_{i=1}^{t-1}{\alpha_i}m^{i}$. The matrix $\tM_t$ is such that its columns form an orthogonal basis for the column space of $M_t$ and $\tM_t^*\tM_t=n\,\identity_{t\times t}$.
Recall that, $\order_t(1)\in\reals^t$ is a finite dimensional random vector that converges to 0
almost surely as $N\to\infty$.

\item[$(b)$] For any pseudo-Lipschitz function $\phi:\reals^{t}\to\reals$ of order $k$,
\begin{align}
\lim_{N\to\infty}\f{1}{N}\sum_{i=1}^N\phi(h_i^2,\ldots,h_i^{t+1})
&\asequal\E\big[\phi(\tau_1 Z_1,\ldots,\tau_tZ_{t})\big]
\end{align}
where $Z_1,\ldots,Z_{t}$ have $\normal(0,1)$ distribution.

\item[$(c)$] For all $1\leq r,s\leq t$ the following equations hold and all limits exist, are bounded and have degenerate distribution
(i.e. they are constant random variables):
\begin{align}
\lim_{N\to\infty}\< h^{r+1},h^{s+1}\>&\asequal\lim_{N\to\infty}\< m^{r},m^{s}\>\,
\end{align}

\item[$(d)$] For all $1\leq r,s\leq t$, and for any
Lipschitz continuous function $\varphi$, the following equations hold and all limits exist, are bounded and have degenerate distribution
(i.e. they are constant random variables):
\begin{align}
\lim_{N\to\infty}\< h^{r+1}, \varphi(h^{s+1})\>&\asequal \lim_{N\to\infty}\< h^{r+1},h^{s+1}\> \< \varphi'(h^{s+1})\>
\end{align}
\item[$(e)$] For $\ell=k-1$, almost surely $\lim_{N\to\infty}(h_i^{t+1})^{2\ell}<\infty$.

\item[$(f)$] For all $0\leq r\leq t$ the following limit exists and there are positive constants $\lbq_r$ (independent of $N$) such that almost surely
\begin{align}
\lim_{N\to\infty}\< m_\perp^r, m_\perp^r\> > \lbq_r\,.
\end{align}

\end{itemize}
\end{lemma}

%
%
\section*{Acknowledgement}

We are deeply indebted  with Erwin Bolthausen, who first
presented the conditioning technique during the EURANDOM
workshop on `Order, disorder and double disorder' in September 2009.
We are also  grateful to Dave Donoho and Arian Maleki,
for a number of insightful exchanges.

This work was partially supported by a Terman fellowship,
the NSF CAREER award CCF-0743978 and the NSF grant DMS-0806211.

%
\appendix

\section{AMP algorithm: An heuristic derivation}
\label{app:AMPderivation}

In this appendix we present an heuristic derivation of the AMP
iteration (\ref{eq:dmm}) starting from the standard message passing
formulation (\ref{eq:mp1}). Let us stress that such derivation is not
relevant for the proof of our Theorem \ref{thm:main}.
Our objective is to help the reader develop an intuitive
understanding of the AMP iteration. For further
discussion of the connection with belief propagation we
refer to \cite{DMM_ITW_I,InPreparationDMM}.

Let us rewrite the message passing iteration for greater convenience of the
reader
\begin{align}
z_{a\to i}^t &= y_a - \sum_{j\in[N]\bs i}A_{aj}x_{j\to a}^t\, ,\label{eq:mp-repeated}\\
x_{i\to a}^{t+1}&=\eta_t\left(\sum_{b\in[n]\bs a}A_{bi}z_{b\to i}^t\right)\, .
\end{align}
Notice that on the right-hand side of both equations the messages appears
in sums of $\Theta(N)$ terms. Consider for instance the  messages
$\{z_{a\to i}^t\}_{i\in [N]}$ for a fixed node $a\in [n]$.
These depend on $i\in [N]$ only because the excluded term
changes. It is therefore natural to guess that
$z^{t}_{a\to i}=z^t_a+O(N^{-1/2})$ and
$x^{t}_{i\to a}=x^t_i+O(n^{-1/2})$, where
$z^t_a$ only depends on the index $a$ (and not on $i$),
and $x^t_i$ only depends on $i$ (and not on $a$).

A na\"{i}ve approximation would consist in neglecting
the $O(N^{-1/2})$ correction
but this turns out to produce a non-vanishing error in the
large-$N$ limit. We instead set
\begin{eqnarray*}
z_{a\to i}^t = z_a^t+\dz_{a\to i}^t\, ,\;\;\;\;\;\;\;
x_{i\to a}^t = x_i^t+\dx_{i\to a}^t\, .
\end{eqnarray*}
Substituting in Eq.~(\ref{eq:mp-repeated}), we get
\begin{align*}
z_{a}^t+\dz_{a\to i}^t &= y_a - \sum_{j\in[N]}A_{aj}
(x_{j}^t+\dx_{j\to a}^t)
+A_{ai}(x_{i}^t+\dx_{i\to a}^t)\, ,\\
x_{i}^{t+1}+\dx_{i\to a}^{t+1}&=
\eta_t\left(\sum_{b\in[n]}A_{bi}(z_{b}^t+\dz_{b\to i}^t)-
A_{ai}(z_{a}^t+\dz_{a\to i}^t)\right)\, .
\end{align*}
We will now drop the terms that are negligible without writing
explicitly the error terms. First of all notice that single terms of
the type $A_{ai}\dz_{a\to i}^t$ are of order $1/N$ and
can be safely neglected. Indeed $\dz_{a\to i} = O(N^{-1/2})$
by our anzatz, and $A_{ai} = O(N^{-1/2})$ by definition.
We get
\begin{align*}
z_{a}^t+\dz_{a\to i}^t &= y_a - \sum_{j\in[N]}A_{aj}
(x_{j}^t+\dx_{j\to a}^t)
+A_{ai}x_{i}^t\, ,\\
x_{i}^{t+1}+\dx_{i\to a}^{t+1}&=
\eta_t\left(\sum_{b\in[n]}A_{bi}(z_{b}^t+\dz_{b\to i}^t)-
A_{ai}z_{a}^t\right)\, .
\end{align*}
We next expand the second equation to linear order in $\dx_{i\to a}^t$ and
$\dz_{a\to i}^t$:
\begin{align*}
z_{a}^t+\dz_{a\to i}^t &= y_a - \sum_{j\in[N]}A_{aj}
(x_{j}^t+\dx_{j\to a}^t)
+A_{ai}x_{i}^t\, ,\\
x_{i}^{t+1}+\dx_{i\to a}^{t+1}&=
\eta_t\left(\sum_{b\in[n]}A_{bi}(z_{b}^t+\dz_{b\to i}^t)\right)-
\eta_t'\left(\sum_{b\in[n]}A_{bi}(z_{b}^t+\dz_{b\to i}^t)\right)
A_{ai}z_{a}^t\, .
\end{align*}
Notice that the last term on the right hand side of
the first equation is the only one dependent on $i$,
and we can therefore identify this term with $\dz_{a\to i}^t$.
We obtain the decomposition
\begin{align}
z_{a}^t &= y_a - \sum_{j\in[N]}A_{aj}
(x_{j}^t+\dx_{j\to a}^t)\, ,\label{eq:Z1}\\
\dz_{a\to i}^t &= A_{ai}x_{i}^t\, .\label{eq:Z2}
\end{align}
Analogously for the second equation we get
\begin{align}
x_{i}^{t+1}&=
\eta_t\left(\sum_{b\in[n]}A_{bi}(z_{b}^t+\dz_{b\to i}^t)\right)\, ,
\label{eq:X1}\\
\dx_{i\to a}^{t+1} & =-
\eta_t'\left(\sum_{b\in[n]}A_{bi}(z_{b}^t+\dz_{b\to i}^t)\right)
A_{ai}z_{a}^t\, .\label{eq:X2}
\end{align}

Substituting Eq.~(\ref{eq:Z2}) in Eq.~(\ref{eq:X1}) to eliminate
$\dz_{b\to i}^t$ we get
\begin{align}
x_{i}^{t+1}&=
\eta_t\left(\sum_{b\in[n]}A_{bi}z_{b}^t+\sum_{b\in[n]}A_{bi}^2x_{i}^t
\right)\, ,
\end{align}
and using the normalization of $A$, we get
$\sum_{b\in [n]}A_{bi}^2\to 1$, whence
\begin{align}
x^{t+1}&=
\eta_t(x^t+A^*z^t)\, .
\end{align}

Analogously substituting Eq.~(\ref{eq:X2}) in (\ref{eq:Z1}),
we get
\begin{align}
z_{a}^t &= y_a - \sum_{j\in[N]}A_{aj}x_{j}^t+
\sum_{j\in[N]}A_{aj}^2\eta_t'(x^t_j+(A^*z^t)_j)z_{a}^t\, .\label{eq:ZZZZ}
\end{align}
Again, using the law of large numbers  and the normalization of $A$,
we get
\begin{align}
\sum_{j\in[N]}A_{aj}^2\eta_t'(x^t_j+(A^*z^t)_j)
\approx \frac{1}{n}\sum_{j\in[N]}\eta_t'(x^t_j+(A^*z^t)_j)
\to \frac{1}{\delta}\<\eta_t'(x^t_j+(A^*z^t)_j)\>\, ,
\end{align}
whence substituting in (\ref{eq:ZZZZ}),
we obtain the second equation in (\ref{eq:dmm}).
This finishes our derivation.
%
%
{
\section{Strong law of large number for triangular arrays}\label{ap:SLLN}
In this section we show how Theorem \ref{thm:SLLN} can be obtained from Theorem 2.1 of Hu and Taylor from \cite{SLLN2}.  Define $a_n\equiv n$, $p\equiv 2$, and $\psi(t)\equiv t^{2+\slln}$. It is clear that $\psi$ satisfies condition (2.1) from \cite{SLLN2}. Next, condition $n^{-1}\sum_{i=1}^n\E|X_{n,i}|^{2+\slln}\le cn^{\slln/2}$ yields
\begin{align}
\sum_{n=1}^\infty \sum_{i=1}^n\frac{\E\{\psi(|X_{n,i}|)\}}{\psi(a_n)} &= \sum_{n=1}^\infty \sum_{i=1}^n\frac{\E\{|X_{n,i}|^{2+\slln}\}}{n^{2+\slln}}\nonumber\\
&\le c\sum_{n=1}^\infty \frac{1}{n^{1+\slln/2}}<\infty\,.\label{eq:appendix-slln-1}
\end{align}
Therefore condition (2.3) from \cite{SLLN2} also holds. Finally, for any positive integer $k$, using condition $n^{-1}\sum_{i=1}^n\E|X_{n,i}|^{2+\slln}\le cn^{\slln/2}$ and a generalized mean inequality
\begin{align*}
\sum_{n=1}^\infty \left(\sum_{i=1}^n\E\bigg(\frac{|X_{n,i}|}{a_n}\bigg)^2\right)^{2k}&= \sum_{n=1}^\infty \left(\sum_{i=1}^n  \frac{\E|X_{n,i}|^{2}}{n^2}\right)^{2k}\\
&\le \sum_{n=1}^\infty \left( \sum_{i=1}^n\frac{\left[\E|X_{n,i}|^{2+\slln}\right]^{\frac{2}{2+\slln}}}{n^2}\right)^{2k} \\
&\le c'+\sum_{n=1}^\infty \left(\sum_{i=1}^n\frac{\E|X_{n,i}|^{2+\slln}}{n^2}\right)^{2k} \\
&\le c'+c''\sum_{n=1}^\infty \frac{1}{n^{2k(1-\frac{\slln}{2})}}<\infty\,.
\end{align*}
The last inequality uses $\slln<1$ which leads to $2k(1-\slln/2)>1$. Hence, condition (2.4) of \cite{SLLN2} satisfies as well. Therefore $n^{-1}\sum_{i=1}^nX_{n,i}$ converges to $0$ almost surely.

%
%
\section{Proof of Lemma \ref{lem:alg-lemma}}\label{ap:alg-lemma}
Define $f(\beta)\equiv\frac{1}{\beta}\log\left(\sum_{i=1}^ne^{\beta \log u_i}\right)$.  Lemma \ref{lem:alg-lemma} is equivalent to show that $f(1+\varepsilon)\le f(1)$.  We prove that $f$ is a decreasing function for all $\beta>0$.  Note that
\begin{align*}
f'(\beta)&=-\frac{1}{\beta^2}\log\left(\sum_{i=1}^ne^{\beta \log u_i}\right)+\frac{1}{\beta}\frac{\sum_{i=1}^n(\log u_i)e^{\beta \log u_i}}{\sum_{s=1}^ne^{\beta \log u_s}}\\
&=-\frac{1}{\beta^2}\sum_{i=1}^n\frac{e^{\beta \log u_i}}{\sum_{s=1}^ne^{\beta \log u_s}}
\left[-\log\left(\sum_{i=1}^ne^{\beta \log u_i}\right)+\beta\log u_i\right]\\
&=-\frac{1}{\beta^2}\sum_{i=1}^n\frac{e^{\beta \log u_i}}{\sum_{s=1}^ne^{\beta \log u_s}}
\log\left(\frac{e^{\beta\log u_i}}{\sum_{i=1}^ne^{\beta \log u_i}} \right)=-\frac{1}{\beta^2}H(p)
\end{align*}
where $H(p)$ is the entropy of a probability distribution on $\{1,\ldots,n\}$ with $p_i=\frac{e^{\beta \log u_i}}{\sum_{s=1}^ne^{\beta \log u_s}}$ and is always non-negative. This finishes the proof.
}
%
%
\section{Proof of probability and linear algebra lemmas}

In this Appendix we provide proofs of two probability
lemmas stated in Section \ref{sec:ProbFacts}.

\subsection{Proof of Lemma \ref{lem:SLLN4us}}
\label{app:SLLN}

Note that by definition of empirical measure, $N^{-1}\sum_{i=1}^N\psi(v_{i})=\E_{\empr_{v}}\{\psi(V)\}$.
The proof uses a truncation technique. For a positive integer $\bound$ define $\psi_\bound$ by
\[
\psi_\bound(x)\equiv\left\{
\begin{array}{ll}
\psi(x)&|\psi(x)|\leq \bound\\
\bound&\psi(x)> \bound\\
-\bound&\psi(x)< -\bound\\
\end{array}
\right.
\]
and write $\psi(x)=\psi_\bound(x)+\tpsi_\bound(x)$. Since $\empr_{v}$ converges weakly to $p_{V}$, for the bounded continuous function $\psi_\bound(x)$,
we have
\begin{align}
\lim_{N\to\infty} \E_{\empr_{v}}\{\psi_\bound(V)\}&=\E_{p_{V}}\{\psi_\bound(V)\}.\label{eq:lim-4-psi-bounded}
\end{align}
On the other hand, since $\psi$ is pseudo-Lipschitz with order $k$
we have $|\psi(x)|\leq L (1+| x|^k)$ for $|x|\ge 1$.
Therefore for large enough $B$,
\[
|\tpsi_\bound(x)|\leq L(1+|x|^k)\indicator_{\{|\psi|>\bound\}}\leq
L(1+|x|^k)\indicator_{\{|x|^k>\f{\bound}{L}-1\}}.
\]
From this we obtain
\begin{multline*}
\E_{p_{V}}\{\psi_\bound(V)\}-\lim\sup_{N\to\infty}\E_{\empr_{v(N)}}\{L
(1+|V|^k)\indicator_{\{|V|^k>\f{\bound}{L}-1\}}\}\\
\leq \lim\inf_{N\to\infty}\E_{\empr_{v(N)}}\{\psi(V)\}\leq\lim\sup_{N\to\infty}\E_{\empr_{v(N)}}\{\psi(V)\}\leq\\
\E_{p_{V}}\{\psi_\bound(V)\}+\lim\sup_{N\to\infty}\E_{\empr_{v(N)}}\{L
(1+|V|^k)\indicator_{\{|V|^k>\f{\bound}{L}-1\}}\}.
\end{multline*}
Now, by assumption $\lim_{N\to\infty}\E_{\empr_{v(N)}}\{|V|^k\}=\E_{p_{V}}\{|V|^k\}$ we can write $|V|^k=|V|^k\indicator_{\{|V|^k>\bound/L-1\}}+|V|^k
\indicator_{\{|V|^k\leq\bound/L-1\}}$ and use the weak convergence of $\empr_{v(N)}$ to $p_{V}$ to get
\[
\lim_{N\to\infty}\E_{\empr_{v(N)}}\{L(1+|V|^k)\indicator_{\{|V|^k\leq
\f{\bound}{L}-1\}}\}=\E_{p_{V}}\{L(1+|V|^k)
\indicator_{\{|V|^k\leq\f{\bound}{L}-1\}}\}.
\]
Therefore
\begin{align*}
\lim\sup_{N\to\infty}\E_{\empr_{v(N)}}\{L(1+|V|^k)\indicator_{\{|V|^k>\f{\bound}{L}-1\}}\}&=\lim_{N\to\infty}
\E_{\empr_{v(N)}}\{L(1+|V|^k)\indicator_{\{|V|^k>\f{\bound}{L}-1\}}\}\\
&=\E_{p_{V}}\{L(1+|V|^k)\indicator_{\{V^k>\f{\bound}{L}-1\}}\}.
\end{align*}
Hence, all we need  to show is that $\E_{p_{V}}\{L|V|^k\indicator_{\{|V|^k>\f{\bound}{L}-1\}}\}$ converges to $0$ as $\bound\to\infty$. But this follows using the bounded $k^{\mbox{\rm \tiny th}}$ moment of $V$ and the dominated convergence theorem, when applied to the sequence of functions $L(1+|V|^k)\indicator_{\{|V|^k>\bound/L-1\}}\leq L(1+|V|^k)$, indexed by $B$.
\endproof
%
%
\subsection{Proof of Lemma \ref{lem:AlmostSmooth}}\label{app:Almost proof}
Recall that by Skorokhod's theorem, there exists a probability
space $(\Omega,\cF,\prob)$ and a construction
of the random variables $\{(X_n,Y_n)\}_{n\ge 1}$ and
$(X,Y)$ on this space, such that letting
\begin{eqnarray*}
A = \big\{\omega\in\Omega\, :\, (X_n(\omega),Y_n(\omega))\to (X(\omega),Y(\omega))\big\}\, ,
\end{eqnarray*}
be the event that $(X_n,Y_n)$ converges to $(X,Y)$,
we have $\prob(A) = 1$.
Let $\cC_F\subseteq\reals^2$ be the domain on which $F$ is continuously
differentiable. Since $F$ is Lipschitz continuous, $\cC_F$
has full Lebesgue measure. Since the probability distribution of $(X,Y)$
is absolutely continuous with respect to Lebesgue,
$\cC_F$ has measure  $1$
under this measure.
Hence if we let
\begin{eqnarray*}
B =\big\{\omega\in\Omega \, : \, (X(\omega),Y(\omega))\in\cC_F\big\}\, ,
\end{eqnarray*}
we have $\prob(B)=1$. On $A\cap B$, we also have
$F'(X_n(\omega),Y_n(\omega))\to F'(X(\omega),Y(\omega))$.

Letting $Z_n(\omega)\equiv F'(X_n(\omega),Y_n(\omega))$
(if $(X_n(\omega),Y_n(\omega))\not\in\cC_F$ set $Z_n(\omega) = 0$) and
$Z(\omega)\equiv F'(X(\omega),Y(\omega))$, we thus proved that
\begin{eqnarray*}
\prob\big\{\lim_{n\to\infty}Z_n(\omega)=Z(\omega)\big\} = 1\, .
\end{eqnarray*}
Since $F$ is Lipschitz $|Z_n(\omega)|\le C$, and hence the bounded convergence theorem
implies $\E\{Z_n(\omega)\}\to \E\{Z(\omega)\}$ which proves our claim.
\endproof

%
%
\subsection{Proof of Lemma \ref{lem:ConditionalVariance}}
\label{app:ConditionalVarianceProof}

Let us
denote by $Q$ the covariance of the gaussian vector
$Z_1,\ldots,Z_t$. The set of matrices $Q$ satisfying the
constraints with constants $c_1,\dots,c_t$, $K$ is compact.
Hence if the thesis does not hold, there must exist a specific
covariance matrix satisfying these constrains, and such that
\begin{eqnarray}
\E\{[\lip(Z_t,Y)]^2\}-u^* C^{-1} u = 0\, .\label{eq:LipSchur}
\end{eqnarray}
Fix $Q$ to be such a matrix, and
let $S\in\reals^{t\times t}$ be the
matrix with entries $S_{i,j}\equiv\E\{
\lip(Z_i,Y)\lip(Z_j,Y)\}$. Then Eq.~(\ref{eq:LipSchur}) implies that
$S$ is not invertible (by Schur complement formula).
Therefore there exist non-vanishing constants $a_1,\dots,a_t$ such that
\begin{eqnarray}
a_1\, \lip(Z_1,Y)+a_2\, \lip(Z_2,Y)+\dots+a_t\,\lip(Z_t,Y) \asequal 0\, .
\end{eqnarray}

The function
$(z_1,\dots,z_t)\mapsto a_1\, \lip(z_1,Y)+\dots+a_t\,\lip(z_t,Y)$
is Lipschitz and non-constant. Hence there is a set ${\cal A}\subseteq
\reals^t$ of positive Lebesgue measure such that it is non-vanishing
on ${\cal A}$. Therefore, ${\cal A}$ must have zero measure under
the law of $(Z_1,\dots,Z_t)$, i.e.
$\lambda_{\rm min}(Q)=0$. This implies that there exists
non-vanishing constants $a_1',\dots,a_t'$ such that
\begin{eqnarray*}
a_1'\, Z_1+a_2'\, Z_2+\dots+a_t'\,Z_t \asequal 0\, .
\end{eqnarray*}
If $t_* = \max\{i\in\{1,\dots,t\}\, :\,  a'_i\neq 0\}$, this implies
\begin{eqnarray*}
Z_{t_*} \asequal \sum_{i=1}^{t_*-1}(-a_i'/a_{t_*}')\, Z_i\, ,
\end{eqnarray*}
which contradicts the assumption $\var(Z_{t_*}|Z_1,\dots,Z_{t_*-1})>0$.
\endproof
%
%
\subsection{Proof of Lemma \ref{lem:InvertibleCovariance}}
\label{app:InvertibleCovariance}

 We will prove the thesis by induction over
$t$. The case $t=1$ is trivial, and assume that the claim
is true up for any $(t-1)$ vectors $v_1,\dots,v_{t-1}$,
with constant $c'_{t-1}$.
Without loss of generality, we will assume $\|v_i\|^2/n\le K$
for some constant $K$ independent of $n$ (increasing the norm
of the $v_i$'s increases $\lambda_{\rm min}(C)$).

Let $V\in\reals^{n\times t}$ be the matrix with columns
$v_1,\dots,v_t$. Then $C = V^*V/n$.
By Gram-Schmidt orthonormalization,
we can construct $A$ upper triangular, and $U\in\reals^{n\times t}$ orthonormal
(i.e., with $U^*U=\identity_{t\times t}$) such that
\begin{eqnarray*}
U = VA\, .
\end{eqnarray*}
It follows that
\begin{eqnarray}
\lambda_{\rm min}(C) = \frac{1}{n}\, \lambda_{\rm min}(V^*V)
=  \frac{1}{n}\, \lambda_{\rm min}((A^{-1})^*A^{-1}) =
 \frac{1}{n}\, \lambda_{\rm max}(AA^{*})^{-1}= \frac{1}{n}\,
\sigma_{\rm max}(A)^{-2} \, .\label{eq:GS}
\end{eqnarray}
Defining $u_i$ to be the columns
of $U$, Gram-Schmidt orthonormalization prescribes
\begin{align*}
u_i & = \frac{v_i-P_{i-1}(v_i)}{\|v_i-P_{i-1}(v_i)\|}
\end{align*}
Which implies $A_{ii}=\|v_i-P_{i-1}(v_i)\|^{-1}\le (cn)^{-1/2}$ and
\begin{align*}
A_{ji} =
-\frac{1}{\|v_i-P_{i-1}(v_i)\|}(\tV_{i-1}^*\tV_{i-1})^{-1}\tV_{i-1}^*v_i
\end{align*}
We then have
\begin{align*}
|A_{ji}| \le (cn)^{-1/2}\lambda_{\rm min}(\tV_{i-1}^*\tV_{i-1})^{-1} (i-1)Kn
\le t(cn)^{-1/2}(c'_{t-1}n)^{-1} Kn\le c''n^{-1/2}\, .
\end{align*}
It follows that $\sigma_{\rm max}(A)\le c'''n^{-1/2}$ (with $c'''$
depending on $n$) whence the thesis follows by Eq.~(\ref{eq:GS}).
\endproof

\bibliographystyle{amsalpha}

\bibliography{all-bibliography}

%
%

\end{document}